\documentclass[a4paper,12pt]{article}

\usepackage{a4wide}
\usepackage{graphicx,subfigure,hhline}
\usepackage[colorlinks=true,linkcolor=blue,citecolor=red]{hyperref}
\usepackage{xcolor}
\usepackage{amssymb,amsthm,empheq,bbold}
\usepackage[inline]{enumitem}
\usepackage[ruled,vlined,linesnumbered]{algorithm2e}
\usepackage{caption}
\theoremstyle{plain} 
\newtheorem{theorem}{Theorem} 
\newtheorem{proposition}{Proposition}
\newtheorem{remark}{Remark}

\usepackage{subfigure,graphicx,hhline,xcolor}
\allowdisplaybreaks
\usepackage{epsfig}
\usepackage{url}
\usepackage{longtable}

\newcommand{\pa}{\partial}

\newcommand{\bx}{\mathbf{x}}
\newcommand{\bk}{\mathbf{k}}

\newcommand{\be}{\begin{equation}}
	\newcommand{\ee}{\end{equation}}
\newcommand{\ben}{\begin{equation}\nonumber}    
	
\begin{document}
		
		\title{Weakly nonlinear analysis of a reaction-diffusion model for demyelinating lesions in Multiple Sclerosis}

		
		\author{Romina Travaglini$^{1,2*}$, Rossella Della Marca$^{3}$\\[1em]
			$^1${\footnotesize INDAM -- National Institute for Advanced Mathematics  ``Francesco Severi''}
			\\{\footnotesize  Piazzale Aldo Moro 5,  00185, Roma, Italy}\\
	       $^2$ {\footnotesize Department of Mathematical, Physical and Computer Science, University of Parma}\\{\footnotesize Parco Area delle Scienze 53/A,
	       	43124, Parma, Italy}
			\\{\footnotesize {travaglini@altamatematica.it} (*corresponding author)}\\
				$^3${\footnotesize Department of Mathematics and Applications ``R. Caccioppoli",}\\ {\footnotesize  University of Naples ``Federico II",}\\ {\footnotesize  Via Cintia, Monte S. Angelo I-80126 Napoli, Italy}  \\{\footnotesize{rossella.dellamarca@unina.it}}}
		%
		%
		
		\maketitle
		
		
		\begin{abstract}
			Multiple Sclerosis is a chronic autoimmune disorder characterized by the degradation of the myelin sheath in the central nervous system, leading to neurological impairments. In this work, we analyze a reaction-diffusion model derived from kinetic theory to study the formation of demyelinating lesions. We perform a Turing instability analysis and a weakly nonlinear analysis to investigate different spatial patterns that may emerge. Our study examines how key parameters, including the squeezing probability of immune cells and the chemotactic response, impact pattern formation. Numerical simulations confirm the analytical results, revealing the emergence of distinct spatial structures.
		\end{abstract}

		\section{Introduction}
		\label{sec:1}
		Multiple Sclerosis (MS) is a chronic neurological disorder in which the immune system mistakenly attacks the brain and spinal cord. It is estimated that over 1.8 million people are affected by MS worldwide \cite{who}.
		The main pathological feature of MS is damage to the myelin sheath, which surrounds axons in the central nervous system and facilitates the transmission of nerve impulses. These lesions can be detected using  magnetic
		resonance imaging, appearing as focal plaques in the white matter. The symptoms of MS involve a gradual worsening of physical and neurological functions. As the disease progresses, individuals may experience muscle weakness, coordination issues, numbness, vision problems, and cognitive changes. These impairments tend to worsen over time as myelin damage disrupts nerve function. The severity and progression of symptoms can vary between individuals, with some individuals experiencing relapses followed by periods of partial recovery, while others may experience a steady decline. Further details about the clinical overview of MS can be found in \cite{lassmann2005multiple,lassmann2007immunopathology,lassmann2012progressive,mahad2015pathological} and references therein.
		\\\indent
		The biological mechanisms underlying MS are mainly recognized as those involving immune cells, which are T-cells, B-cells, macrophages, and microglia. These cells can become dysfunctionally activated against myelin-producing cells (oligodendrocytes) or myelin itself, triggering an autoimmune cascade. This inflammatory process is driven by the production of proinflammatory cytokines, molecules produced by immune cells, able to attract themselves and stimulate their clonal expansion.
		\\\indent
		It is important to note
		that self-reactive immune cells can also be found in non-pathological conditions \cite{danke2004autoreactive}. In such cases, the action of specific cells, known as immunosuppressors, can inhibit or eliminate activated immune cells or those cells presenting the antigen able to activate them. What is observed in autoimmune conditions like MS is that the number and effectiveness of these natural killers are compromised \cite{hoglund2013one,mimpen2020natural,zozulya2008role}.
		\\\indent
		Mathematical modeling represents a powerful framework for understanding the underlying dynamics of MS. On one hand, different models, based on ordinary differential equations \cite{frascoli2022dynamics}, stochastic dynamics \cite{bordi2013mechanistic}, and systems of partial differential equations (PDE) \cite{barresi2016wavefront,bisi2024chemotaxis,calvez2008mathematical,khonsari2007origins,lombardo2017demyelination}, have been developed to replicate key pathological features.
		In particular, the systems of PDE proposed have been extensively analyzed using Turing instability and weakly nonlinear analysis. These studies have led to the identification of spatial patterns that mimic myelin lesions typically seen in MS, specifically for type III lesions and Balo's concentric lesions. These lesions are characterized by extensive oligodendrocyte loss, limited T-cell presence, and a lack of remyelination.
		\\\indent
		On the other hand, the kinetic theory of active particles has been revealed to be effectively applicable to the autoimmune process modeling \cite{delitala2013mathematical,kolev2018mathematical,costa2021optimal,della2022mathematical,menale2024nonconservative,ramos2019kinetic}. This approach has enabled the derivation of macroscopic descriptions of key biological phenomena by starting from the mesoscopic level of cell interactions and cytokine-induced migration \cite{Oliveira2022}. It has been adapted to the specific case of MS \cite{travaglini2023reaction}, offering insights into inflammation, lesion development, and disease progression in the presence of the remyelination process, also including treatment terms \cite{travaglini2023}. The reaction-diffusion models derived in these works are analyzed through a Turing instability analysis to inquire about pattern formation.

		Our present aim is to advance the study of demyelinating lesion formation by performing a weakly nonlinear analysis of a reduced reaction-diffusion system derived from the model introduced in \cite{travaglini2023reaction}, and initially explored in one dimension. In contrast to other models that treat oligodendrocyte loss as an indirect proxy for demyelination \cite{bisi2024chemotaxis,bisi2025derivation,lombardo2017demyelination}, the formulation presented in \cite{travaglini2023reaction} incorporates an explicit state variable representing the fraction of destroyed myelin, enabling a direct and biologically precise description of demyelinating lesions. Although a two-dimensional analysis of a related model was carried out in \cite{travaglini2023}, that study documented pattern emergence without employing weakly nonlinear techniques. Here, we extend that framework by modeling explicit myelin destruction to two dimensions and applying a weakly nonlinear analysis. Thereby, we establish a novel analytical approach to investigate how key mechanistic parameters--such as squeezing probability and chemotactic sensitivity--govern the spatial organization of demyelinated areas. This methodology captures structured patterns reminiscent of histopathological findings in MS, ranging from elongated plaques \cite{multz2025multiple, peterson2001transected} to focal or concentric lesions \cite{stadelmann2005tissue, tzanetakos2020heterogeneity}. Consequently, it offers new insights into the emergence of atypical lesion morphologies that cannot be captured in one-dimensional models or in two-dimensional simulations alone without weakly nonlinear analysis.
		\\\indent
		The chapter is organized as follows: in Section \ref{sec:2}, we recall the reaction-diffusion model for MS proposed in \cite{travaglini2023reaction} and perform a reduction of it. The pattern formation is inquired in Section \ref{sec:3}, by performing Turing and weakly nonlinear analysis. Section \ref{sec:4} shows numerical simulations that confirm the results obtained analytically. In Section \ref{sec:5}, some concluding remarks are provided.
		
		\section{The model}
		\label{sec:2}

		Let us start by recalling the macroscopic reaction-diffusion model derived in  \cite{travaglini2023reaction}. The model describes the formation and restoration of myelin plaques in Multiple Sclerosis and has the following shape:
		\be\label{eq:macA}
		\begin{aligned}
			\dfrac{\pa \,A(t,\bx)}{\pa t} \,&=\, \alpha + p_{AR} \,A(t,\bx) \,R(t,\bx) - d_{AS} \,A(t,\bx) \,S(t,\bx) - d_A \,A(t,\bx), 
			\\[3mm]
			\dfrac{\pa \,S(t,\bx)}{\pa t} \,&=\, p_{SA} \,S(t,\bx) \,A(t,\bx) - d_S \,S(t,\bx),
			\\[3mm]
			\dfrac{\pa \,R(t,\bx)}{\pa t} \,&=\,
			\nabla_{\bx}\,\cdot\left[D_R\,\phi_0(\,R(t,\bx))\nabla_{\bx}\,\,\,R(t,\bx) - \chi \,\phi_1(\,R(t,\bx))\,\,R(t,\bx)  \,\nabla_{\bx}\,C\right]\\ 
			&\qquad+ p_{RA} \,R(t,\bx) \,A(t,\bx) - d_{RS} \,R(t,\bx) \,S(t,\bx) - d_R \,R(t,\bx),
			\\[3mm]
			\dfrac{\pa C(t,\bx)}{\pa t}\,&=\,D_C\, \Delta_{\bx}\,C(t,\bx) +p_{C}\,A(\bx,t)\,R(\bx,t)- d_C C(\bx,t),
			\\[3mm]
			\dfrac{\pa E(t,\bx)}{\pa t}\,&=\,(\bar E - E(t,\bx))\dfrac{\,d_{ER}\,R(t,\bx)}{\tilde r+R(t,\bx)}\,R(t,\bx)-r_E\,E(t,\bx).
		\end{aligned}
		\ee
		
		The involved quantities are macroscopic densities for:
		self-antigen presenting cells ($A$) that can activate the self-reactive immune cells against the myelin; immunosuppressive cells ($S$) that can inhibit or eliminate activated both self-antigen presenting cells and immune cells; self-reactive immune cells ($R$); cytokines ($C$) produced by immune cells when activated by the self-antigen presenting cells; and destroyed fraction of myelin ($E$).
		The variables depend on $t\in\mathbb{R}_0^+$ and $\bx\in\Gamma_{\bx}$, where $\Gamma_{\bx}$ is a bounded domain in $\mathbb{R}^2$. We underline that the kinetic description provided in \cite{travaglini2023reaction} and in the previous works like \cite{della2022mathematical,ramos2019kinetic} also include conservative dynamics, whose effects are not observable at the macroscopic level. System \eqref{eq:macA}, indeed, accounts only for the non-conservative processes, which are described below. 
		\\\indent
		The proliferative processes participating in the evolution of $A$ are represented by a constant input source $\alpha$ depending on external factors \cite{selmi2018long} and proliferative interactions with $R$. The proliferation of $S$ and $R$ comes from interactions with $A$, while $S$ destructively interacts with $A$ and $R$. The production of cytokines comes from the interaction between $A$ and $R$. Lastly, the consumption of myelin by $R$ is coupled with its restoration by oligodendrocytes, and the last equation in \eqref{eq:macA} comes from the different time scales considered in \cite{travaglini2023reaction}. Here, $\bar E$ represents the total amount of myelin, which we consider constant in space and time. We also take into account the natural decay of cell populations and cytokines.   We take as positive constants all the proliferative and destructive interaction rates, $p_{ij}$ and $d_{ij}$, respectively, the decay rates $d_i$,  with $i\in\{A,S,R,C,E\}$ and $j\in\{A,S,R\}$, the cytokines production rate $p_C$,  the saturation constant involved in demyelination $\tilde r$,  and the remyelination rate  $r_E$.
		\\\indent
		Moreover, model \eqref{eq:macA}  incorporates: the diffusive dynamics for $R$, whose coefficient is the constant $D_R$ multiplied by the function $\phi_0$ depending on the density; the diffusive dynamics for $C$, which has a constant diffusion coefficient  $D_C$; the chemotactic motion of $R$ towards regions where the  concentration
		of cytokines is higher, with a sensitivity function $\phi_1$ 
		and the maximal chemotactic rate $\chi$.  \\ 
		The functions $\phi_0$ and $\phi_1$ are taken  as follows \cite{wang2007classical}
		\be\label{FunPhi01}
		\phi_1(y)=\left(1-\left(\frac{y}{\bar R}\right)^\gamma\right)\mathbf 1_{[0,\bar R]}, \quad\phi_0(y)=\phi_1(y)-y\,\phi_1'(y). 
		\ee
		The function $\phi_1$ is referred to as the \textit{squeezing probability}, i.e., the probability of a cell finding space at its
		neighboring site. Typically, this probability decreases linearly with the cell density at that site. Here, we consider a rather general form for $\phi_1$ by assuming that the \textit{squeezing exponent} $\gamma$ in \eqref{FunPhi01} can be greater than or equal to one. The term $\mathbf 1$ is the characteristic function and $\bar R$ is the maximal density of self-reactive immune cells. 
		We report in Table \ref{Tab:param} all the variables and parameters of model (\ref{eq:macA})-(\ref{FunPhi01}) along with their description.

		\begin{longtable}{ll}
			\caption{Description of variables and parameters relevant to model
				(\ref{eq:macA})--(\ref{FunPhi01}).}
			\label{Tab:param}\\
			\hline
			Symbol & Description \\
			\hline
			\endfirsthead
			
			\hline
			Symbol & Description \\
			\hline
			\endhead
			
			$A$ & Macroscopic density for self-antigen presenting cells\\
			$S$ & Macroscopic density for immunosuppressive cells\\
			$R$ & Macroscopic density for self-reactive immune cells\\
			$C$ & Macroscopic density for cytokines\\
			$E$ & Destroyed fraction of myelin\\
			$\alpha$ & Input source of self-antigen presenting cells\\
			$p_{AR}$ & Proliferative rate of $A$ due to interactions with $R$\\
			$p_{SA}$ & Proliferative rate of $S$ due to interactions with $A$\\
			$p_{RA}$ & Proliferative rate of $R$ due to interactions with $A$\\
			$p_C$ & Production rate of cytokines\\
			$d_{AS}$ & Destructive rate of $A$ due to interactions with $S$\\
			$d_{RS}$ & Destructive rate of $R$ due to interactions with $S$\\
			$d_{ER}$ & Destructive rate of $E$ due to interactions with $R$\\
			$d_A$ & Natural decay rate of self-antigen presenting cells\\
			$d_S$ & Natural decay rate of immunosuppressive cells\\
			$d_R$ & Natural decay rate of self-reactive immune cells\\
			$d_C$ & Natural decay rate of cytokines\\
			$D_R$ & Diffusion coefficient for self-reactive immune cells\\
			$D_C$ & Diffusion coefficient for cytokines\\
			$\chi$ & Maximal chemotactic rate for self-reactive immune cells\\
			$\gamma$ & Squeezing exponent\\
			$\bar R$ & Maximal density of self-reactive immune cells\\
			$\bar E$ & Density of total myelin (both sane and destroyed)\\
			$\tilde r$ & Saturation parameter for demyelination\\
			$r_E$ & Remyelination rate\\
			\hline
		\end{longtable}
		
		The purpose of the present work is to investigate the formation of lesions in the myelin sheath, caused by the inflammation which is, in turn, induced by the interplay between self-reactive immune cells and cytokines.
		To this aim, we assume that the populations of  self-antigen-presenting cells and immunosuppressive cells are in an equilibrium state, i.e., their temporal derivatives are zero.
		Consequently, we may derive a reduced form of the system \eqref{eq:macA}, as stated in the following proposition.
		\begin{proposition}
			Let us suppose that the time derivatives of the two populations of self-antigen presenting cells $A(t,\bx)$ and immunosuppressive cells $S(t,\bx)$ are zero for all $\bx\in\Gamma_{\bx}$ and that the two populations are not constantly equal to zero. Then, the system \eqref{eq:macA} can be reduced to three equations for the quantities $R(t,\bx)$, $C(t,\bx)$, and $E(t,\bx)$, and rewritten in the dimensionless form 
			\be\label{eq:rdsR}
			\begin{aligned}
				\frac{\pa R}{\pa t}=\,&\nabla_{\bx}\cdot\left(\Phi_0(R)\,\nabla_{\bx}\, R-{\xi} \,\Phi_1(R)\,\nabla_{\bx}\,C\right)+R\,(1-R) , \\[2mm]
				\frac{\pa C}{\pa t}=\,&\delta\Delta_{\bx}\,C+\beta\,R-\tau\,C,\\[2mm]
				\frac{\pa E}{\pa t}=\,&\frac{\theta\,R^2}{\Omega+R}\left(1-E\right)-\zeta\,E.
			\end{aligned}
			\ee
		\end{proposition}
		\begin{proof}
			By equalizing to zero the first two equations of system \eqref{eq:macA}, we get
			\be\label{ASMac}
			A(t,\bx)=\frac{d_S}{p_{SA}} ,\qquad
			S(t,\bx)= \frac{\alpha\, p_{SA}-d_A d_S}{d_S\, d_{AS}} + \frac{p_{AR}}{d_{AS}}\,R(t,\bx),
			\ee
			by assuming that $S(t,\bx)\neq 0$ and $A(t,\bx)\neq 0$.
			We report here that the relations above may be obtained starting from the kinetic description and considering different time scales for the different processes. Nevertheless, we leave this computation for future work. 
			By plugging the expressions \eqref{ASMac} into  \eqref{eq:macA},
			we get
			
			\begin{align}
				\dfrac{\pa \,A(t,\bx)}{\pa t} \,&=0, 
				\\[3mm]\nonumber
				\dfrac{\pa \,S(t,\bx)}{\pa t} \,&=0,
				\\[3mm]\nonumber
				\dfrac{\pa \,R(t,\bx)}{\pa t} \,&=\,	\nabla_{\bx}\,\cdot\left[D_R\,\phi_0(R(t,\bx))\,\nabla_{\bx}\,R(t,\bx) - \chi \,\phi_1(R(t,\bx))\,R(t,\bx)  \,\nabla_{\bx}\,C\right]\\ \nonumber
				&\qquad  - d_{RS} \,R(t,\bx) \,\left(\frac{\alpha\, p_{SA}-d_A d_S}{d_S\, d_{AS}} + \frac{p_{AR}}{d_{AS}}\,R(t,\bx)\right)\\ \nonumber
				&\qquad  + R(t,\bx)\left(p_{RA}  \,\frac{d_S}{p_{SA}} - d_R \right),
				\\[3mm]\nonumber
				\dfrac{\pa C(t,\bx)}{\pa t}\,&=\,D_C\, \Delta_{\bx}\,C(t,\bx) +p_{C}\,\frac{d_S}{p_{SA}}\,R(\bx,t)- d_C C(\bx,t),
				\\[3mm]\nonumber
				\dfrac{\pa E(t,\bx)}{\pa t}\,&=\,(\bar E - E(t,\bx))\dfrac{\,d_{ER}\,R(t,\bx)}{\tilde r+R(t,\bx)}\,R(t,\bx)-r_E\,E(t,\bx).
			\end{align}
			
			Thus, defining the coefficients 
			\ben
			\lambda={\frac{p_{RA}d_S}{p_{SA}}+d_{RS}\frac{ d_A\,d_S-\alpha p_{SA}}{d_S\, d_{AS}}}-d_R,\quad \mu=\frac{p_{AR}\,d_{RS}}{d_{AS}},\quad
			\psi=\dfrac{p_{C}\,d_S}{p_{SA}},
			\ee
			we obtain a system of equations, uncoupled from the first two, for the densities $R$, $C$, $E$ 
			\be\label{eqMmac}	
			\begin{aligned}
				\frac{\pa R(t,\bx)}{\pa t} &=
				\nabla_{\bx}\,\cdot\left[D_{ R}\,\phi_0(R(t,\bx))\nabla_{\bx}\,\,R(t,\bx) - \chi \,\phi_1(R(t,\bx))\,R(t,\bx)  \,\nabla_{\bx}\,C\right]\\ &
				\qquad + R(t,\bx) \left(\lambda-\mu R(t,\bx)\right)\\
				\frac{\pa C(t,\bx)}{\pa t}&=D_C\, \Delta_{\bx}\,C(t,\bx) + \psi\,  R(\bx,t)- d_C C(\bx,t),
				\\
				\dfrac{\pa E(t,\bx)}{\pa t}\,&=\,(\bar E - E(t,\bx))\dfrac{\,d_{ER}\,R(t,\bx)}{\tilde r+\,R(t,\bx)}\,R(t,\bx)-r_EE(t,\bx).
			\end{aligned}
			\ee
			At this point, we find it convenient to perform the change of variables
			\ben \tilde t=\lambda\, t,\qquad \tilde{\bx}=\sqrt{\frac{\lambda}{D_R}}\bx,\ee 
			in this way, the partial derivatives are
			\ben
			\frac{\partial}{\partial t} = \lambda \, \frac{\partial}{\partial \tilde{t}}, 
			\qquad 
			\nabla_{\bx} = \sqrt{\frac{\lambda}{D_R}} \, \nabla_{\tilde{\bx}}, 
			\qquad 
			\Delta_{\bx} = \frac{\lambda}{D_R} \, \Delta_{\tilde{\bx}},
			\ee
			so that \eqref{eqMmac} becomes (omitting the dependence on space and time) 
			\ben
			\begin{aligned}
				\frac{\partial R}{\partial \tilde t} &= 
				\nabla_{\tilde \bx} \cdot \left[ 
				\phi_0(R)\, \nabla_{\tilde \bx} R 
				- \frac{\chi}{D_R}\, \phi_1(R)\, R\, \nabla_{\tilde\bx}  C 
				\right]
				+ R \left( 1 - \frac{\mu}{\lambda} R \right), \\[6pt]
				\frac{\partial C}{\partial \tilde t} &= 
				\frac{D_C}{D_R}\, \Delta_{\tilde \bx} C 
				+ \frac{\psi}{\lambda} R 
				- \frac{d_C}{\lambda} C, \\[6pt]
				\frac{\partial E}{\partial \tilde t} &= 
				\frac{d_{ER}}{\lambda}\, (\bar{E} - E)\, 
				\frac{R^2}{\tilde{r} + R} 
				- \frac{r_E}{\lambda} E,
			\end{aligned}
			\ee
			where we have divided both sides of each equation by $\lambda$.
			The system above can now be written in the dimensionless form by setting
			\ben 
			\tilde R=\frac{\mu}{\lambda}\,{R},\quad
			\tilde C=\frac{\mu\,d_C}{\lambda\,\bar\psi}\,{C},\quad
			\tilde E=\frac{E}{\bar E},
			\ee
			with $\bar\psi$ the maximal cytokines production rate for $R$-cell, providing
			\ben
			\begin{aligned}
				\frac{ \partial\tilde R}{\partial \tilde t} &= 
				\nabla_{\tilde \bx} \cdot \left[ 
				\phi_0\left(\frac{\lambda}{\mu}\tilde R\right)\, \nabla_{\tilde \bx} \tilde R 
				- \frac{\chi}{D_R}\,\,\frac{\psi}{d_C}\, \phi_1\left(\frac{\lambda}{\mu}\tilde R\right)\, \tilde R\, \nabla_{\tilde\bx} \tilde C 
				\right]
				+ \tilde R \left( 1 - \tilde R \right), \\[6pt]
				\frac{\partial\tilde C}{\partial \tilde t} &= 
				\frac{D_C}{D_R}\, \Delta_{\tilde \bx}\tilde C 
				+ \frac{\psi\,d_C}{\bar\psi\,\lambda} \tilde R 
				- \frac{d_C}{\lambda} \tilde C, \\[6pt]
				\frac{\partial \tilde  E}{\partial \tilde t} &= 
				\frac{d_{ER}}{\mu}\, (1 - \tilde  E)\, 
				\frac{\tilde R^2}{\tilde{r}\dfrac{\mu}{\lambda} + \tilde R} 
				- \frac{r_E}{\lambda} \tilde E.
			\end{aligned}
			\ee
			We can now define the new coefficients of the model {as}
			$$
			\xi=\chi\,\frac{\bar \psi}{D_R d_C\,},\quad
			\delta=\frac{D_C}{D_M},\quad
			\beta=\frac{\psi\,d_C}{\bar \psi\,\lambda},$$
			\ben\tau=\frac{d_C}{\lambda},\quad
			\theta=\frac{d_{ER}}{\mu},\quad\Omega=\frac{\tilde r\,\mu}{\lambda},\quad\zeta=\frac{r_E}{\lambda},
			\ee
			moreover, we have 
			\ben
			\phi_1\left(\frac{\lambda}{\mu}\tilde R\right)=\left(1-\left(\frac{\lambda}{\mu}\frac{\tilde R}{\bar R}\right)^\gamma\right)\mathbf 1_{U}(1-\gamma),
			\ee
			where $U$ is the interval $\left\{0\leq {\lambda \tilde R}/{\mu}\leq\bar R\right\}$,
			and 
			\ben
			\phi_0\left(\frac{\lambda}{\mu}\tilde R\right)=\phi_1\left(\frac{\lambda}{\mu}\tilde R\right)-\frac{\lambda}{\mu}\tilde R\,\frac{d\,\phi_1}{d\,R}\left(\frac{\lambda}{\mu}\tilde R\right)=\phi_1\left(\frac{\lambda}{\mu}\tilde R\right)-\tilde R\frac{d\,\phi_1}{d\,\tilde R}\left(\frac{\lambda}{\mu}\tilde R\right).
			\ee
			Thus, we introduce the quantity $\mathcal{R}={\mu \bar R}/{\lambda}$ and functions 
			\be\label{FunPPhi01}
			\Phi_1( y)=\left(1-\left(\frac{y}{\mathcal R}\right)^\gamma\right)\mathbf 1_{[0,\mathcal R]}(1-\gamma), \quad\Phi_0(y)=\Phi_1(y)-y\,\Phi_1'(y).
			\ee
			At this point, by renaming the non-dimensional variables and densities removing the tilde, we get the system in \eqref{eq:rdsR}, and this completes the proof.
			\begin{flushright}$\square$\end{flushright}
		\end{proof}

		\section{Pattern formation analysis}
		\label{sec:3}
		
		In this section, we investigate the properties of the parameters of system \eqref{eq:rdsR} to lead to the formation of spatial patterns. To this aim, we add to system \eqref{eq:rdsR} the non-negative initial data:
		\ben
		\mathbf{W}(0,\bx)=\mathbf{W}_0(\bx)\geq 0, \mbox{ with } \mathbf{W} = (R, C, E) \mbox{ and } R_0\leq \mathcal R,\,E_0<1.
		\ee
		We do not consider a maximum level of cytokines, as their production can vary considerably throughout the course of Multiple Sclerosis \cite{Amoriello2024,Kamma2022}. We also impose zero-flux conditions at the boundary:
		\be\label{NoFlux}
		\left(\Phi_0(R)\,\nabla_{\bx}\, R-{\xi} \,\Phi_1(R) R\,\nabla_{\bx}\,C \right)\cdot {\bf \widehat n}=0,\quad \nabla_{\bx} C\cdot {\bf \widehat n}=0,
		\ee
		being  ${\bf \widehat n}$ the external unit normal to {the boundary} $\pa\Gamma_{\bx}$. 
		We remark that zero-flux boundary conditions are used to prevent external inputs, allowing patterns to emerge intrinsically from the system. This aligns with classical Turing instability results in biological systems \cite{murray2003mathematical}, where spatial structures arise from internal dynamics rather than boundary effects. From a biological point of view, this is reasonable because, in the brain, the blood–brain barrier effectively isolates the tissue from the external environment. In the context of Multiple Sclerosis, we assume that autoreactive lymphocytes have already infiltrated the central nervous system through the barrier, so the subsequent evolution of demyelinating plaques occurs within this closed environment.
		
		\subsection{Turing instability}
		\label{subsec:2.2}
		We want to investigate the capability of the model to describe the formation of plaques in the brain constituted by destroyed myelin. Hence, we perform a Turing instability analysis \cite{turing52} of system \eqref{eq:rdsR} with boundary conditions \eqref{NoFlux}.
		
		The necessary conditions for Turing instability to occur are that the system without spatial gradients has a spatially homogeneous steady state, which turns unstable when adding diffusive and chemotactic terms.

		By straightforward computations, we may detect the only nonzero equilibrium of system \eqref{eq:rdsR}, given by
		$${\bf W}_1= \left(1,\,\frac{\beta}{\tau},\,\frac{\theta}{\theta+\zeta\,(1+\Omega)}\right),$$ which is admissible provided that $\mathcal R>1$.
		By linearizing the system  \eqref{eq:rdsR} around the equilibrium $\bf W_1$, we have
		\be\label{SistW}
		\displaystyle \frac{\pa{\bf W}}{\pa t}={\mathbb D}\,\Delta_{\bf x}{\bf W}+{\mathbb J}\,{\bf W} \; \mbox{on} \; (0,\infty)\times\Gamma_{\bx},
		\ee
		the diffusion matrix $\mathbb D$
		and the
		Jacobian $\mathbb J$ write
		\ben 
		{\mathbb D} = \left(
		\begin{array}{ccc}
			\tilde\Phi_0& -\xi\, \tilde\Phi_1& 0 \\[1mm]
			0 & \delta & 0 \\[1mm]
			0 & 0 & 0
		\end{array}
		\right),\quad  \tilde\Phi_l=\Phi_l(1),\,l=0,1,
		\ee
		and
		\ben
		\mathbb J=
		\begin{pmatrix}
			-1 & 0 & 0 \\
			\beta & -\tau & 0 \\
			\dfrac{\zeta\, \theta\, (1 + 2 \Omega)}{(1 + \Omega)(\zeta + \theta + \zeta\, \Omega)} & 0 & -\dfrac{\zeta + \theta + \zeta\, \Omega}{1 + \Omega}
		\end{pmatrix},
		\ee
		respectively.
		It is immediate to check that the equilibrium $\bf W_1$ is locally asymptotically stable in spatially homogeneous conditions, i.e., when the diffusive and chemotactic processes do not come into play in system \eqref{eq:rdsR}.
		Indeed, the characteristic polynomial of the Jacobian is
		$$
		P(\lambda)=-(1 + \lambda)(\lambda + \tau)\left(\lambda+ \frac{\zeta + \theta + \zeta \, \Omega}{1 + \Omega}\right),
		$$
		so that the eigenvalues are $
		\left\{ -1,  -\tau, -({\zeta + \theta + \zeta \, \Omega})/({1 + \Omega})\right\}.
		$
		Since all eigenvalues have negative real parts, by the standard linearization theory for ordinary differential equations (see, e.g., \cite{arnold1992ordinary}), the equilibrium ${\bf W_1}$ is asymptotically stable.

		Considering the full system \eqref{SistW}, we expand the solution in a Fourier series:
		\be\label{Wgen}
		{\bf W}({\bf x},t) = \sum_k c_k \, e^{\lambda_k t} \, \overline{\bf W}_k({\bf x}),
		\ee
		where the eigenfunctions $\overline{\bf W}_k({\bf x})$ solve the time-independent problem
		\ben
		\begin{cases}
			\Delta_{\bf x} \overline{\bf W}_k + k^2 \overline{\bf W}_k = {\bf 0}, & \text{on } \Gamma_{\bf x},\\[1mm]
			{\bf \hat n} \cdot \nabla_{\bf x} \overline{\bf W}_k = 0, & \text{on } \partial \Gamma_{\bf x}.
		\end{cases}
		\ee
		Substituting \eqref{Wgen} into \eqref{SistW} shows that, for each wavenumber $k$, the growth rate $\lambda_k$ is an eigenvalue of the matrix
		\[
		\mathbb J - k^2 \mathbb D.
		\]
		
		Thus, instability occurs if there exists a range of wavenumbers $[k_1,k_2]$ such that $\operatorname{Re}(\lambda_k) > 0$ for $k \in [k_1,k_2]$. For our system, the eigenvalues of $\mathbb J - k^2 \mathbb D$ are
		\begin{align*}
			\lambda_k^0 &= - \frac{\zeta + \theta + \zeta \, \Omega}{1 + \Omega},\\[1mm]
			\lambda_k^{\pm} &= \frac{1}{2} \left( -1 - \tau - k^2 (\delta + \tilde\Phi_0)  \pm \sqrt{\left( -1 - \tau - k^2 (\delta + \tilde\Phi_0)\right)^2-4\,h(k^2)} \right),
		\end{align*}
		where $$h(k^2)=-\dfrac{\left(1 + \Omega\right)\,\mbox{Det}\left(\mathbb J - k^2 \mathbb D\right)}{\zeta + \theta + \zeta\, \Omega}=k^4 \delta\, \tilde\Phi_0 + k^2 \Lambda + \tau, \quad 
		\Lambda = \delta + \tau\, \tilde\Phi_0 - \beta\, \xi\,\tilde\Phi_1.$$
		Since $\lambda_k^0 < 0$, instability can only arise from $\lambda_k^{\pm}$. This requires $h(k^2)<0$ for some $k^2$ and corresponds to imposing 
		$$
		\Lambda < 0 \quad \text{and} \quad \Lambda^2 - 4 \delta\,\tilde\Phi_0\, \tau > 0.
		$$
		The first condition gives
		$
		\xi > ({\delta + \tau \, \tilde{\Phi}_0})/({\beta \, \tilde{\Phi}_1}),
		$  
		while the second condition leads to two possibilities:
		$$
		\xi < \frac{\delta + \tau \, \tilde{\Phi}_0 - 2 \sqrt{\delta\, \tau\, \tilde{\Phi}_0} }{\beta \, \tilde{\Phi}_1}
		\quad \text{or} \quad
		\xi > \frac{\delta + \tau \, \tilde{\Phi}_0 + 2 \sqrt{\delta\, \tau\, \tilde{\Phi}_0}}{\beta \, \tilde{\Phi}_1}.
		$$
		This leads to the Turing instability condition
		\be\label{TurCond}
		\xi > \xi_c := \frac{\delta + \tau\, \tilde\Phi_0+2 \sqrt{\delta\, \tau \tilde\Phi_0} }{\beta\, \tilde\Phi_1}.
		\ee
		The corresponding critical wavenumber such that $h(k^2)=0$ is
		\be\label{CappaCri}
		k_c^2 = \frac{\sqrt{\Lambda ^2-4 \delta \,\tau \,\tilde\Phi_0}-\Lambda_c }{2 \delta \,\tilde\Phi_0}= \sqrt{\frac{\tau}{\delta\, \tilde\Phi_0}},\quad \mbox{ where } \Lambda_c= \delta + \tau\, \tilde\Phi_0 - \beta\, \xi_c\,\tilde\Phi_1.
		\ee

		\subsection{Weakly nonlinear analysis}\label{SubWeak}
		
		Once the conditions on parameters leading to the formation of spatial patterns are outlined, we want to reproduce a richer scenario beyond what Turing analysis alone can achieve. Specifically, we perform a two-dimensional weakly nonlinear analysis of the problem, that relies on the fact that, when the control parameter $\xi$ approaches its critical threshold, the system’s dynamics evolve more slowly. This property makes it possible to study pattern formation through the use of amplitude equations.  
		
		More precisely, as discussed for instance in \cite{walgraef2012spatio}, each possible {steady-state configuration} of the reaction--diffusion system corresponds to a {planform} defined by $m$ pairs of wave vectors $({\bf k}_j , -{\bf k}_j)$, For {critical modes}, those satisfying $|{\bf k}_j| = k_c$, the solution can be written as 
		\be\label{SolAj}
		\widetilde{\bf U} = \sum_{j=1}^m \left[ {\bf A}_j(t) \, e^{i\,{\bf k}_j \cdot \bx} + \overline{{\bf A}}_j(t) \, e^{-i\,{\bf k}_j \cdot \bx} \right],
		\ee
		where ${\bf A}_j$ denotes the {complex amplitude vector} associated with the mode ${\bf k}_j$, and $\overline{{\bf A}}_j$ is its {complex conjugate}.  
		Depending on the value of $m$ and the geometric relations among the wave vectors, different {pattern symmetries} can be investigated. In this work, we focus on the case $m = 2$.
		
		Our aim is, then, to recover an analytical shape for the amplitudes ${\bf A}_j$ when the system evolves from a perturbation of the equilibrium state, thus, we start by
		performing a Taylor expansion of system  \eqref{eq:rdsR} up to the third
		order around the equilibrium $\bf W_1$, getting
		\be\label{SistComp1}
		\begin{aligned}
			\frac{\partial {\bf U}}{\partial t} &= \mathcal L\,{\bf U} + \mathcal H[{\bf U}]+O({\bf U}^3), \quad \mbox{for} \quad
			{\bf U}=\left(\begin{array}{c}u\\v\\w\end{array}\right)=\bf W-\bf W_1,
		\end{aligned}
		\ee
		where $\mathcal L  ={\mathbb D}\Delta_{\bf x}+ {\mathbb J}$ and $\mathcal H[{\bf U}]$ is
		\ben
		\begin{pmatrix}
			\displaystyle-u^2 +\left(\tilde{\Phi}_0'+\dfrac12\tilde{\Phi}_0''u\right)\,u\,\Delta_{\bf x}u +\left(\tilde{\Phi}_0'+\tilde{\Phi}_0''u\right)\,||\nabla_{\bf x} u||^2+\\-\xi\,\left[\left(\tilde{\Phi}_1'+\dfrac12\tilde{\Phi}_1''u\right)\,u\,\Delta_{\bf x}v +\left(\tilde{\Phi}_1'+\tilde{\Phi}_1''u\right)\,\nabla_{\bf x} u\cdot\nabla_{\bf x} v\right] \\[5mm]
			0\\[5mm]
			\dfrac{\theta}{\left(1+\Omega\right)^2}\left( \dfrac{2 \, u^2 \, \zeta \, \Omega^2}{\theta + \zeta \, \left(1 + \Omega\right)}- u \, w \, \left(1 + 2 \, \Omega\right)+\right.\\\left. 
			- \dfrac{u^3 \, \zeta \, \Omega^2}{\left(1 + \Omega\right) \left(\theta + \zeta \, \left(1 + \Omega\right)\right)} -\dfrac{u^2 \, w \, \Omega^2}{1 + \Omega} 
			\right)
		\end{pmatrix}, 
		\ee
		where $\tilde\Phi_l'=\Phi_l'(1), \tilde\Phi_l''=\Phi_l''(1),\,l=0,1$.

		As mentioned before, we can gather information on the formation and shape of patterns when the parameter $\xi$ is close to the critical threshold $\xi_c$ given in \eqref{TurCond}, thus we expand it in terms of a small parameter $\eta$ as follows
		\be\label{expxi}
		\xi=  \xi_c +\eta\,\xi_1+\eta^2\,\xi_2+\eta^3\,\xi_3+O(\eta^3).
		\ee
		Similarly, the solution vector ${\bf U}$ can be expanded as
		\be\label{UExp}
		{\bf U} = \eta
		\left(
		\begin{array}{c}u_1 \\ v_1\\w_1 \end{array}
		\right)
		+ \eta^2 
		\left(
		\begin{array}{c}u_2 \\ v_2\\w_2 \end{array}
		\right)
		+ \eta^3
		\left(
		\begin{array}{c}u_3 \\ v_3\\w_3 \end{array}
		\right)
		+ O(\eta^3).
		\ee
		When the bifurcation parameter approaches the threshold, the pattern's amplitude changes slowly over time, since the 
		leading term of the nonlinear expansion of the solution is the product of the basic pattern and a slowly varying amplitude \cite{hoyle2006pattern,vanhecke1994amplitude}. This enables a distinction between fast and slow temporal dynamics, 
		allowing us to introduce multiple time scales as follows:
		\be\label{timeExp}
		\frac{\partial}{\partial t} = \eta \frac{\partial}{\partial T_1} + \eta^2 \frac{\partial}{\partial T_2} + O(\eta^2). 
		\ee
		We can now substitute the expansions \eqref{UExp}-\eqref{timeExp} in \eqref{SistComp1},  and collecting the terms in the same order of $\eta$ on each side, we obtain three equations, one for each order:
		\begin{itemize}
			\item[-] order $\eta$:
			\be\label{SistOrd1}
			\mathcal L_c
			\left(
			\begin{array}{c}u_1 \\ v_1\\w_1 \end{array}
			\right)
			=0\quad
			\mbox{ 
				with 
			}
			\mathcal L_c  = \mathbb D_c\, \Delta_{\bf x}+\mathbb J,\quad \mathbb D_c=\begin{pmatrix}
				\tilde\Phi_0 &-\xi_c\,\tilde\Phi_1& 0\\[2mm]
				0 &  \delta\,&0\\[2mm]
				0  &0  & 0
			\end{pmatrix}
			\ee
			\item[-] order $\eta^2$:
			\be\label{EqOrd2}
			\frac{\pa}{\pa T_1}\left(
			\begin{array}{c}u_1 \\ v_1\\w_1 \end{array}
			\right)=	
			\mathcal L_c
			\left(
			\begin{array}{c}u_2 \\ v_2\\w_2 \end{array}
			\right)
			+ \mathcal H_2\left[\left(
			\begin{array}{c}u_1 \\ v_1\\w_1 \end{array}
			\right)\right]
			\ee
			with 
			$$
			\mathcal H_2\left[\left(
			\begin{array}{c}u_1 \\ v_1\\w_1 \end{array}
			\right)\right]=$$ 	\ben
			\left(
			\begin{array}{c}\tilde{\Phi}_0'\, \nabla_{\bf x}\cdot\left(u_1 \nabla_{\bf x} u_1\right)-\xi_c\, \tilde{\Phi}_1' \nabla_{\bf x}\cdot\left(u_1 \nabla_{\bf x} v_1\right)-\xi_1\,d_{12}\Delta_{\bf x}v_1- u_1^2 
				\\[2mm] 0\\[2mm]\dfrac{\theta}{\left(1+\Omega\right)^2}\left( \dfrac{2 \, u_1^2 \, \zeta \, \Omega^2}{\theta + \zeta \, \left(1 + \Omega\right)}- u_1 \, w_1 \, \left(1 + 2 \, \Omega\right) \right)\end{array}
			\right)
			\ee
			\item[-] order $\eta^3$:
			\be\label{EqOrd3}
			\frac{\pa}{\pa T_1}\left(
			\begin{array}{c}u_2 \\ v_2\\w_2 \end{array}
			\right)+	
			\frac{\pa}{\pa T_2}\left(
			\begin{array}{c}u_1 \\ v_1\\w_1 \end{array}
			\right)=	
			\mathcal L_c
			\left(
			\begin{array}{c}u_3 \\ v_3\\w_3 \end{array}
			\right)
			+ \mathcal H_3\left[\left(
			\begin{array}{c}u_1 \\ v_1\\w_1 \end{array}
			\right),\left(
			\begin{array}{c}u_2 \\ v_2\\w_2 \end{array}
			\right)\right]
			\ee
			with
			$$
			\mathcal H_3\left[\left(
			\begin{array}{c}u_1 \\  v_1\\w_1 \end{array}
			\right),\left(
			\begin{array}{c}u_2 \\v_2\\w_2 \end{array}
			\right)\right]=
			$$
			\ben
			\left(
			\begin{array}{c}\tilde{\Phi}_0'\, \nabla_{\bf x}\cdot\left(u_1 \nabla_{\bf x} u_2+ u_2 \nabla_{\bf x} u_1\right)+\\-
				\xi_c\left[\tilde{\Phi}_1'\, \nabla_{\bf x}\cdot\left(u_1 \nabla_{\bf x} v_2+ u_2 \nabla_{\bf x} v_1\right)+\dfrac12\tilde{\Phi}_1''\, \nabla_{\bf x}\cdot\left(u_1^2 \nabla_{\bf x} v_1 \right)\right]+\\[2mm]+\dfrac12\tilde{\Phi}_0''\, \nabla_{\bf x}\cdot\left(u_1^2 \nabla_{\bf x} u_1 \right)-\xi_1\, \tilde{\Phi}_1' \nabla_{\bf x}\cdot\left(u_1 \nabla_{\bf x} v_1\right)+\\-\xi_1\,d_{12}\Delta_{\bf x}v_2-\xi_2\,d_{12}\Delta_{\bf x}v_1-u_1\,u_2\\[5mm] 0\\[5mm] \dfrac{\theta}{\left(1+\Omega\right)^2}\left( \dfrac{4 \, u_1\,u_2 \, \zeta \, \Omega^2}{\theta + \zeta \, \left(1 + \Omega\right)}\right.+\\\left.- \left(u_1\,w_2+u_2\,w_1\right) \left(1 + 2 \, \Omega\right) 
				- \dfrac{u_1^3 \, \zeta \, \Omega^2}{\left(1 + \Omega\right) \left(\theta + \zeta \, \left(1 + \Omega\right)\right)} -\dfrac{u_1^2 \, w_1 \, \Omega^2}{1 + \Omega} 
				\right) \end{array}
			\right).
			\ee
		\end{itemize}
		
		Each one of the equalities above can provide information on the corresponding term of the expansion \eqref{Expu1} or its time derivative. We start from the equation of order $\eta$ and claim the following result.
		\begin{proposition}\label{Prop2}
			The solution of \eqref{SistOrd1} can be expressed as
			\be\label{Expu1}
			\left(
			\begin{array}{c}u_1\\v_1\\w_1 \end{array}
			\right)=\left(
			\begin{array}{c}\rho_R\\1\\\rho_E \end{array}
			\right)\sum_{j=1}^2\left[{\mathcal W}_j(t)\,e^{i\,{\bf k_j\cdot\bx}}+\overline{\mathcal W}_j(t)\,e^{-i\,{\bf k_j\cdot\bx}}\right],
			\ee
			where $\overline{\mathcal W}_j$ is the complex conjugate,
			\be\label{rhoR}
			\rho_R=\dfrac{\sqrt{\tau}}{\beta} \, \left(\sqrt{\tau} + \sqrt{\dfrac{\delta}{\tilde\Phi_0}}\right),
			\ee
			and 
			\be \label{rhoE}
			\rho_E=\dfrac{\zeta \, \theta \, \sqrt{\tau}\left(1 + 2 \, \Omega\right)}{\beta \, \, \left(\theta + \zeta \, \left(1 + \Omega\right)\right)^2} \, \left(\sqrt{\tau} + \sqrt{\dfrac{\delta}{\tilde\Phi_0}}\right).
			\ee
		\end{proposition}
		\proof See Appendix \ref{App1}.

		Successively, we look at the equation of order $\eta^2$, i.e. \eqref{EqOrd2}. We can provide the following result.
		\begin{proposition}\label{Prop3}
			Let us consider the equation \eqref{EqOrd2} and let the vector  $\left(u_1,\\ v_1,\\ w_1\right)^T$ be in the form \eqref{Expu1}; then, the equation \eqref{EqOrd2} can be solved if and only if the conditions 
			\be\label{Solv1}
			\left(\rho_R+\sigma_C\right)\frac{\pa \mathcal W_j}{\pa T_1}=
			\xi_1\,\,k_c^2\,\tilde{\Phi}_1\,\mathcal W_j,\quad j=1,2, 
			\ee hold, and the solution has the form
			\be\label{Expu2}
			\begin{aligned}
				\left(
				\begin{array}{c}u_2 \\ v_2\\w_2 \end{array}
				\right)
				=& \left(
				\begin{array}{c}X_0 \\ Y_0\\Z_0 \end{array}
				\right)\left(|\mathcal W_1|^2+|\mathcal W_2|^2\right)+\sum_{j=1}^2 \left(
				\begin{array}{c}\rho_R \\ 1\\\rho_E \end{array}
				\right)\,\mathcal V_j\,e^{i\,\bk_j\cdot\bx}
				\\
				&+ \left(
				\begin{array}{c}X_{1} \\ Y_{1}\\Z_{1}\end{array}
				\right)\,\left[\mathcal W_1\, \mathcal W_2\,e^{i\,(\bk_1+\bk_2)\cdot\bx}
				+\mathcal W_1\,\overline{ \mathcal W}_2\,e^{i\,(\bk_1-\bk_2)\cdot\bx}\right]
				\\
				&+\sum_{j=1}^2 \left(
				\begin{array}{c}X_{2} \\ Y_{2}\\Z_{2}\end{array}
				\right)\,\mathcal W_j^2\,e^{2\,i\,\bk_j\cdot\bx}+\mbox{ c.c.},
			\end{aligned}
			\ee 
			where c.c. indicates the complex conjugate, and being the coefficients $X_m, Y_m, Z_m$, $m=0,1,2$, given by
			\be\label{CoefX0}
			\left(
			\begin{array}{c}X_0 \\ Y_0\\Z_0 \end{array}
			\right)
			=
			-2\,\rho_R^2	\left(
			\begin{array}{c}1 \\[3mm] \dfrac{\beta}{\tau} \\[3mm]\Xi_0
			\end{array}
			\right),
			\ee
			\be\label{CoefX1}
			\left(
			\begin{array}{c}X_1 \\ Y_1\\Z_1\end{array}
			\right)=
			\dfrac{ \rho_R\left( \rho_R + k_c^2 \left( \rho_R \tilde{\Phi}_0' - \xi_c  \tilde{\Phi}_1' \right) \right) }{1 + 2 k_c^2 \tilde{\Phi}_0}\left(
			\begin{array}{c}-2 \, 
				\\[4mm]\dfrac{\beta }{k_c^2 \xi_c \tilde{\Phi}_1}
				\\[4mm] \Xi_1
			\end{array}
			\right),
			\ee
			\be\label{CoefX2}
			\begin{aligned}
				\left(
				\begin{array}{c}X_2 \\ Y_2\\Z_2 \end{array}
				\right)
				=& 	\dfrac{-\rho_R\left(\rho_R + 2 \, k_c^2 \, \left(\rho_R \, \tilde{\Phi}_0' - \xi_c \, \tilde{\Phi}_1'\right)\right)}{\tau + 4 \, k_c^2 \, \left(\delta +\tilde{\Phi}_0\,\left( 4 \, k_c^2 \, \delta  + \tau \right) - \beta \, \xi_c \, \tilde{\Phi}_1\right)}
				\,\left(
				\begin{array}{c}
					4 \, k_c^2 \, \delta + \tau
					\\[2mm]
					\beta
					\\[2mm]
					\Xi_2
				\end{array}
				\right),
			\end{aligned}
			\ee
			with $\Xi_0$, $\Xi_1$ and $\Xi_2$ provided in \ref{App2}.
		\end{proposition}
		\proof See Appendix \ref{App2}.
		
		For what concerns the equation \eqref{EqOrd3} of order $\eta^3$, it leads to the following result.
		\begin{proposition}\label{Prop4}
			Let us consider the equation \eqref{EqOrd3} and let the vectors\\  $\left(u_1,\\ v_1,\\ w_1\right)^T$, $\left(u_2,\\ v_2,\\ w_2\right)^T$, be in the form \eqref{Expu1} and \eqref{Expu2}, respectively; then, the equation \eqref{EqOrd3} can be solved if and only if the conditions below are satisfied.
			\be\label{Solv2}
			(\rho_R + \sigma_C)\left(\frac{\partial \mathcal V_j}{\partial T_1} + \frac{\partial \mathcal W_j}{\partial T_2}\right) =\, 
			k_c^2\,\tilde{\Phi}_1\,\left (\xi_2 \mathcal W_j + \xi_1\,\mathcal V_j\right)
			+\left(r_1\,|\mathcal W_j|^2 + r_2\,|\mathcal W_l|^2 \right)\,\mathcal W_j,
			\ee
			for $j,l=1,2$, $j\neq l$, with coefficients
			\be\label{r1r2}
			\begin{aligned}
				r_1=&-(X_0 + X_2)\, \rho_R \, \left(2 + k_c^2\, \tilde{\Phi}_0'\right)\\[2mm]&
				+ k_c^2\, ( X_0-X_2)\, \xi_c\, \tilde{\Phi}_1'
				+ \frac{1}{2}\, k_c^2\, \rho_R\, \left( \xi_c\, \left(4\, Y_2\, \tilde{\Phi}_1' + \rho_R\, \tilde{\Phi}_{1}''\right)-\rho_R^2\, \tilde{\Phi}_{0}''\right)
				,\\[2mm]
				r_2=&-(2\, X_1 + X_0)\, \rho_R\, \left(2 + k_c^2\, \tilde{\Phi}_0'\right)
				+ k_c^2\, \xi_c\, \left(X_0 + 2\, X_1\, \rho_R\right)\, \tilde{\Phi}_1'
				\\[2mm]&- k_c^2\, \rho_R^2\, \left(\rho_R\, \tilde{\Phi}_{0s} - \xi_c\, \tilde{\Phi}_{1s}\right).
			\end{aligned}
			\ee
		\end{proposition}
		\proof See Appendix \ref{App3}.
		
		Now, we look at the expansion \eqref{UExp} and, keeping in mind the expressions \eqref{Expu1} and \eqref{Expu2}, we may assert that the emerging solution given in \eqref{SolAj} (with $m=2$) can be recovered by defining the amplitudes of patterns associated with modes ${\bf k}_j$ as \be\label{ExpAj}
		{\bf A}_j=\eta\, \left(
		\begin{array}{c}\rho_R \\ 1\\\rho_E\end{array}
		\right)\,\mathcal W_j+\eta^2\, \left(
		\begin{array}{c}\rho_R \\ 1\\\rho_E\end{array}
		\right)\,\mathcal V_j+O(\eta^3),\quad j=1,2,
		\ee
		{and their evolution in time, considering the expansion \eqref{timeExp}, is given by
			\be\label{dAjdt}
			\frac{d {\bf A}_j}{d t}=\eta^2\, \left(
			\begin{array}{c}\rho_R \\ 1\\\rho_E\end{array}
			\right)\,\frac{\pa\mathcal W_j}{{\pa T_1}}+\eta^3\, \left(
			\begin{array}{c}\rho_R \\ 1\\\rho_E\end{array}
			\right)\,\left(\frac{\pa\mathcal W_j}{{\pa T_2}}+\frac{\pa\mathcal V_j}{{\pa T_1}}\right)+O(\eta^4),\quad j=1,2,
			\ee}
		We also find it convenient to write amplitudes as ${\bf A}_j=\left(\rho_R\,A_j,\,A_j,\rho_E\,A_j\right)^T$. At this point, we outline the result for the evolution in time of amplitudes.
		\begin{proposition}\label{Prop5}
			Let us decompose each $A_j(t)$ into its mode and phase angle as $A_j(t)=\rho_j(t)\,e^{i\,\phi_j(t)}$, for $j=1,2$. Then, $\phi_1$ and $\phi_2$ are constant and the evolution in time of $\rho_1,\,\rho_2$ is prescribed by the system
			\be\label{SistRhoPhi}
			\begin{aligned}		
				s_0\,\frac{d \rho_1}{d t} &= \xi_m\,\rho_1 +s_1 \, \rho_1^3 +s_2\,\rho_2^2\, \rho_1
				\\
				s_0\,\frac{d \rho_2}{d t} &= \xi_m \,\rho_2 +s_1 \, \rho_2^3 +s_2\,\rho_1^2  \, \rho_2,
			\end{aligned}
			\ee
			where 
			\ben s_0=\dfrac{\rho_R+\sigma_C}{k_c^2\,\xi_c\,\tilde{\Phi}_1},\quad \xi_m=\dfrac{\xi-\xi_c}{\xi_c},\quad s_i=\dfrac{r_i}{k_c^2\,\xi_c\,\tilde{\Phi}_1}, \ee
			with $i=1,\,2$.
		\end{proposition}
		\proof See Appendix \ref{App4}

		The existence and stability results relevant to the emergence of different spatial patterns are given by the following theorem.
		\begin{theorem}\label{Teo}
			The system \eqref{SistRhoPhi} admits the following stationary  states that correspond to the different spatial patterns:
			\begin{itemize}
				\item[i)]  homogeneous solution with $\rho_1=\rho_2=0$; in this case, no pattern  emerges;
				\item[ii)]  equilibria $\mathcal S_1=\left(\bar\rho,\,0\right)$ and $\mathcal S_2=\left(0,\,\bar\rho\right)$, with $\bar\rho=\sqrt{-{\xi_m}/{s_1}}$, corresponding to striped pattern, that exists for $s_1<0$ and are stable for $s_0 > 0$ and $s_2 < s_1 < 0$;
				\item[iii)]   equilibrium $\mathcal Q=\left(\tilde\rho,\,\tilde\rho\right)$, with $\tilde\rho=\sqrt{-{\xi_m}/(s_1+s_2)},$  corresponding to squared pattern, that exists for $s_1<-s_2$ and is stable for $s_0 > 0$ and  $s_1 < -|s_2|$. 
			\end{itemize}	
			
		\end{theorem}
		\begin{proof}
			The conditions for the existence of both striped and squared patterns are directly provided by observing that $\xi_m>0$ from \eqref{TurCond}.  In order to inquire about the stability of patterns, let us compute the Jacobian matrix of the system \eqref{SistRhoPhi} evaluated at the equilibria. For $\mathcal S_1$ we have
			\ben
			\boldsymbol{\mathcal{J}}|_{\mathcal S_1}=
			\begin{pmatrix}
				-\dfrac{2 \, \xi_m}{s_0} & 0 \\
				0 & \dfrac{(s_1 - s_2) \, \xi_m}{s_0 \, s_1}
			\end{pmatrix},
			\ee
			having trace and determinant  
			\ben
			Tr\,\boldsymbol{\mathcal{J}}|_{\mathcal S_1}=-\dfrac{(s_1 + s_2) \, \xi_m}{s_0 \, s_1}
			,\quad Det\,\boldsymbol{\mathcal{J}}|_{\mathcal S_1}=\dfrac{2 \, (-s_1 + s_2) \, \xi_m^2}{s_0^2 \, s_1}.
			\ee
			It is easy to check that $Tr\,\boldsymbol{\mathcal{J}}|_{\mathcal S_1}=Tr\,\boldsymbol{\mathcal{J}}|_{S_2}$ and $Det\,\boldsymbol{\mathcal{J}}|_{\mathcal S_1}=Det\,\boldsymbol{\mathcal{J}}|_{S_2}$. Then, the conditions for local asymptotic stability of the striped pattern providing $
			Tr\,\boldsymbol{\mathcal{J}}|_{\mathcal S_1}<0$ and $Det\,\boldsymbol{\mathcal{J}}|_{\mathcal S_1}>0$ are
			\be\label{StabStri}    
			s_0 > 0, \quad  s_2 < s_1 < 0.
			\ee
			
			For the stability of the squared pattern, instead, we compute the Jacobian matrix evaluated at the equilibrium $\mathcal Q$, obtaining 
			
			\ben
			\boldsymbol{\mathcal{J}}|_{\mathcal Q}=-\dfrac{2 \,  \xi_m}{s_0 \, (s_1 + s_2)}
			\begin{pmatrix}
				s_1 && s_2 \\
				\\
				s_2 && s_1
			\end{pmatrix},
			\ee
			whose trace and determinant  are
			\ben
			Tr\,\boldsymbol{\mathcal{J}}|_{\mathcal Q}=-\dfrac{4 \, s_1 \, \xi_m}{s_0 \, (s_1 + s_2)}
			,\quad Det\,\boldsymbol{\mathcal{J}}|_{\mathcal Q}=\dfrac{4 \, (s_1 - s_2) \, \xi_m^2}{s_0^2 \, (s_1 + s_2)}.
			\ee
			In this case, the conditions for local asymptotic stability of the squared pattern are
			\be\label{StabSqua}    
			s_0 > 0,\quad  s_1 < -|s_2|.
			\ee
			\begin{flushright}$\square$\end{flushright}
		\end{proof}
		
		\begin{remark}
			The spatial patterns predicted by Theorem \ref{Teo} correspond to structured distributions of demyelinated regions that may resemble histopathological observations in Multiple Sclerosis. Specifically, striped or band‑like configurations can be interpreted as elongated plaques often seen along white‑matter tracts, where demyelinated areas follow the orientation of deep medullary veins or fiber bundles, a pattern classically exemplified by “Dawson’s fingers” \cite{multz2025multiple}; additionally, evidence from \cite{peterson2001transected} shows elongated plaques even in lesions confined to the cortex. Some of these striped patterns may also represent individual segments of the concentric rings seen in Baló’s lesions \cite{stadelmann2005tissue}, while squared or spotted arrangements reflect focal or concentric demyelinating areas reminiscent of either intracortical lesions \cite{peterson2001transected} or the fully developed rings characteristic of Baló’s sclerosis \cite{tzanetakos2020heterogeneity}.\end{remark}

		\section{Numerical simulations}
		\label{sec:4}
		
		In this section, we numerically investigate the different possible behaviors of the reaction-diffusion system \eqref{eq:rdsR} by varying the parameter settings. The aim is to get patterning results in accordance with the theoretical findings outlined in the previous section. First of all, we fix the following  parameter values:
		\be
		\begin{array}{c}
			\beta = 1,\quad \tau = 5,\quad \theta = 1,\quad \Omega = 1,\quad \zeta =0.2.\label{Pars}
		\end{array}
		\ee
		Then, we analyze the dynamics of the model for different choices of the squeezing exponent $\gamma$, defining functions \eqref{FunPhi01}. In particular, we take $\gamma=1,2$, and we investigate the values for the parameters $\mathcal R$ (given in \eqref{FunPPhi01}) and $\delta$ which determine the existence and stability of patterns in the two cases. The corresponding results are reported in  Figure \ref{Bif}.
		
		\begin{remark}
			We point out here that the values for parameters relevant to cytokine dynamics and chemotactic motion are taken accordingly to the ones used in other mathematical models for Multiple Sclerosis, e.g. \cite{lombardo2017demyelination}, and are biologically plausible. Conversely, the parameter values governing the progression of myelin destruction and the range for ${R}$ are chosen strictly for illustrative purposes here, with the intention that more specific estimations will be addressed in future investigations.\end{remark}
		
		\begin{figure}[ht!]
			\centering
			\begin{tabular}{cc}
				(a)&(b)\\
				\includegraphics[width=0.35\linewidth]{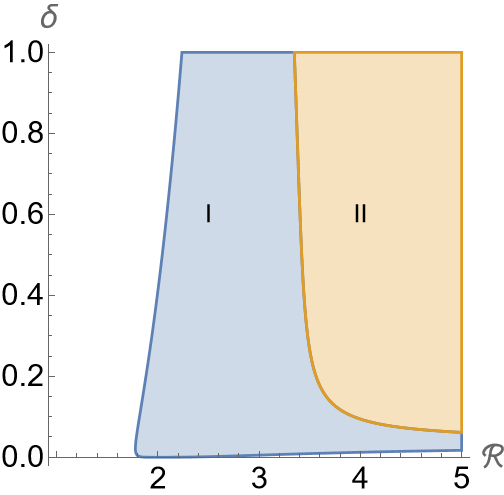}
				&
				\includegraphics[width=0.35\linewidth]{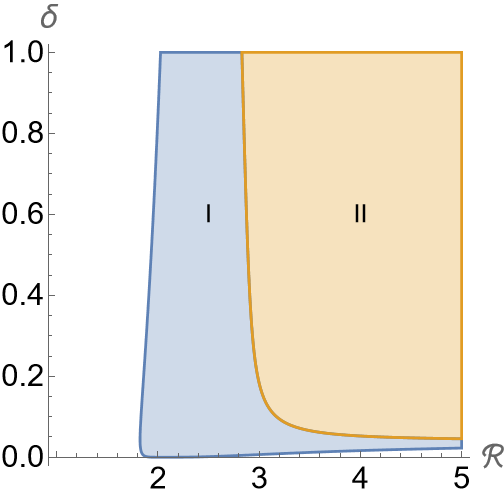}
			\end{tabular}
			\caption{Values for $\mathcal R$ and $\delta$ satisfying condition  \eqref{StabStri} (region I) or condition \eqref{StabSqua} (region II), taking values as in \eqref{Pars} and fixing $\gamma=1$ in panel (a)  and  $\gamma=2$ in panel (b).}
			\label{Bif}
		\end{figure}
		In Figure \ref{Bif}, values for $\mathcal R$ and $\delta$ leading to stable stripes as stated by condition \eqref{StabStri} are in region I, while values providing stable squares stated by condition \eqref{StabSqua}  are in region II, the two panels correspond to the cases $\gamma=1$ (panel (a)) and $\gamma=2$ (panel (b)).

		Then, we depict the behavior of quantities involved in system \eqref{eq:rdsR}-\eqref{NoFlux}. We take as reference case the one with $\gamma=1$, $\delta=0.6$ and  $\mathcal R=3.2$. From Figure \ref{Bif}, we observe that the stability of striped patterns is expected. Additionally, we take $\xi_m=0.01$ that, with this choice of parameters, provides a value for the chemotactic sensitivity  $\xi\simeq 13.31$.  We perform numerical simulations using the software \texttt{VisualPDE} \cite{walker2023visualpde} in a square domain of size $L=5\pi$. The results are shown in Figure \ref{Cases1}, where the portion of myelin destroyed is reported. 
		\begin{figure}[ht!]
			\centering
			\begin{tabular}{c}
				{\quad\includegraphics[width=0.35\textwidth]{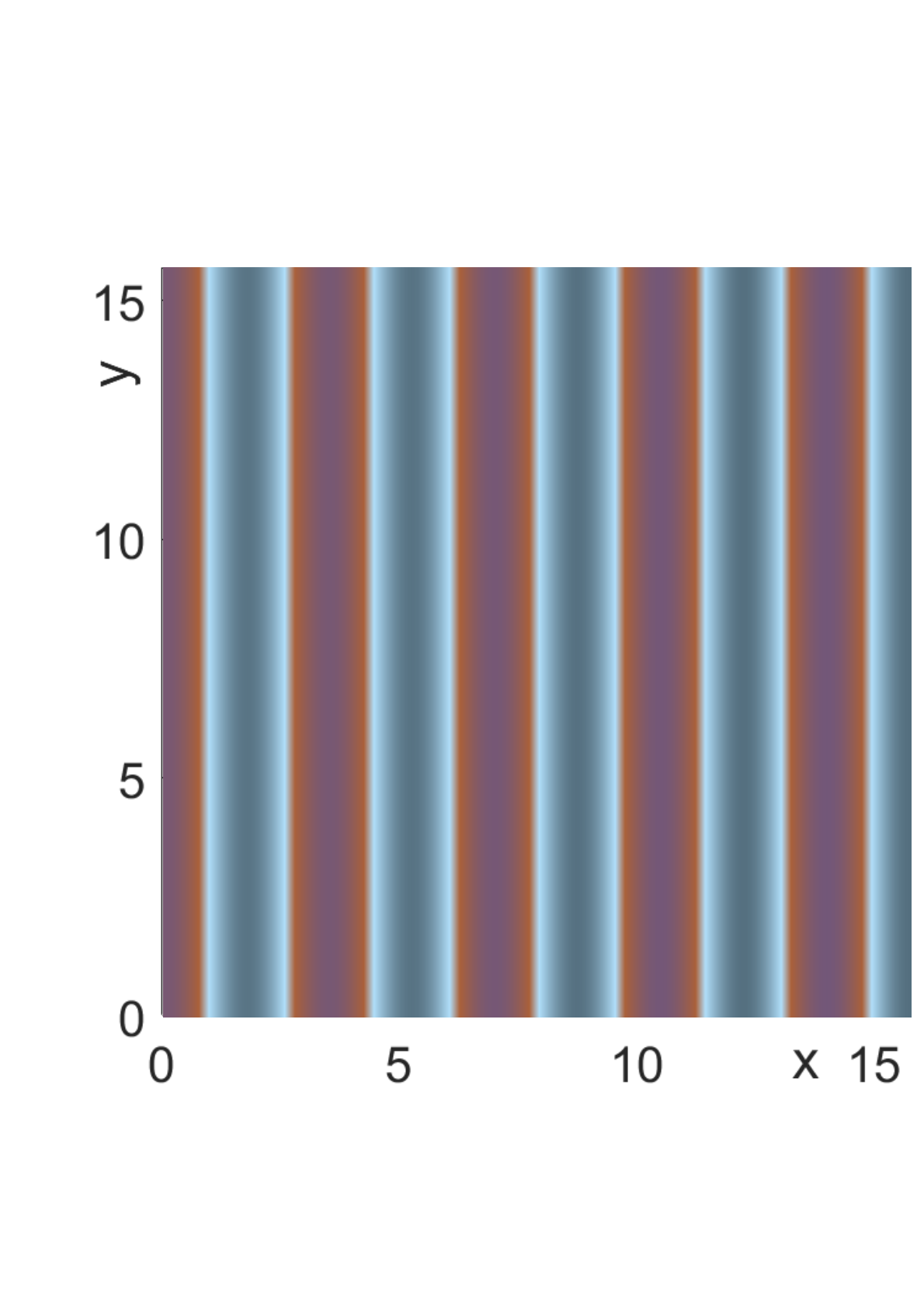}
					\qquad\includegraphics[width=0.4\textwidth]{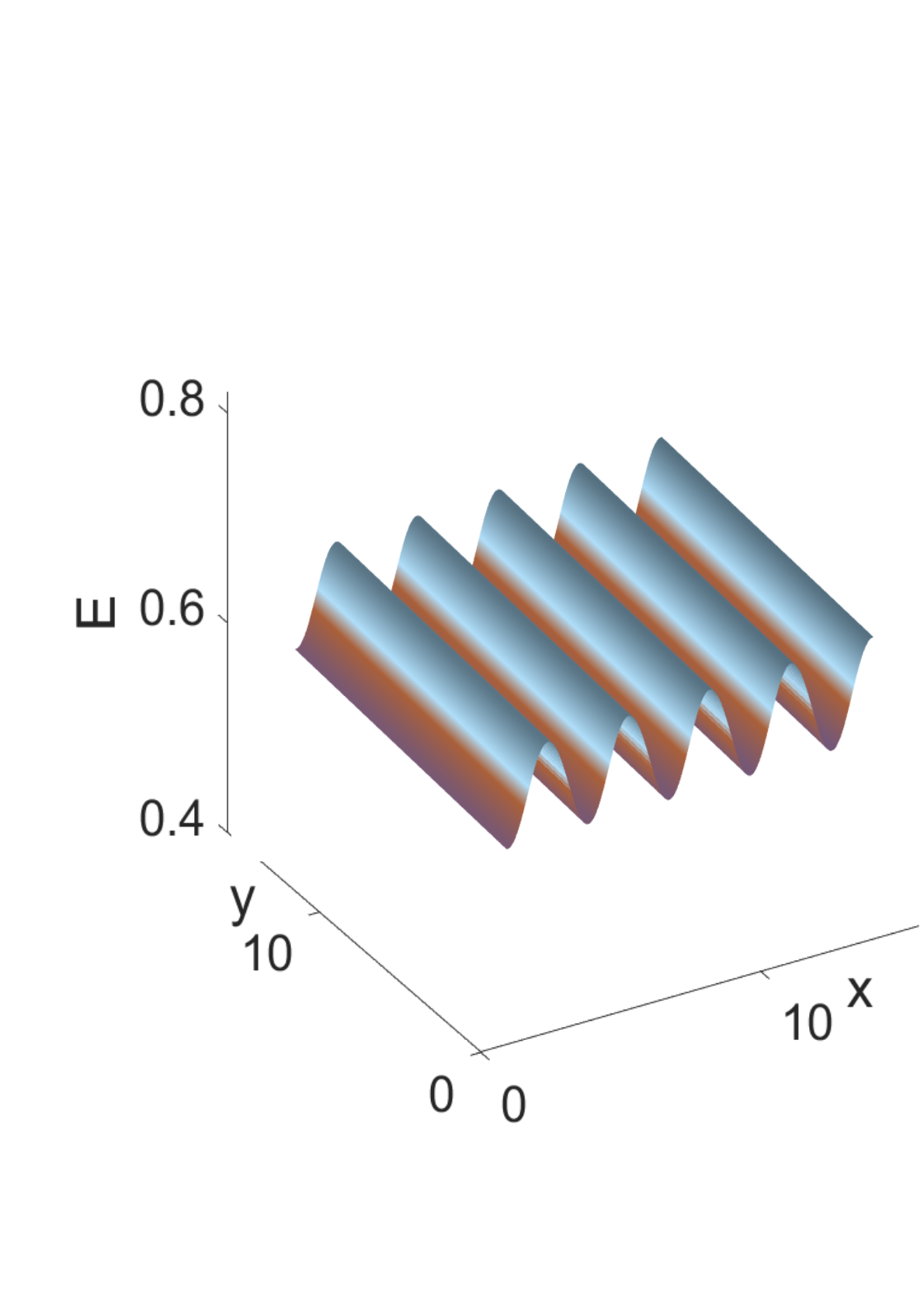}}\\
				\includegraphics[width=0.6\textwidth]{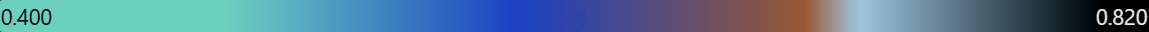}
			\end{tabular}
			\caption{ Long-time patterning of destroyed myelin portion ($E$), described by system \eqref{eq:rdsR}, taking parameters  $\gamma=1$, $\delta=0.6$ and  $\mathcal R=3.2$, $\xi=13.31$. }\label{Cases1}
		\end{figure}
		
		As a second case, we keep the same values of parameters as before, except for $\gamma=2$, providing a value for the chemotactic sensitivity  $\xi\simeq 10.88$. In this case, as observable in Figure \ref{Bif} again, a squared symmetry of the pattern is the stable configuration. The results are shown in Figure \ref{Cases2}, where we observe a spotted spatial configuration. Indeed, in terms of modeling, a higher squeezing exponent leads to a higher diffusive term and a lower chemotactic term for self-reactive immune cells, which means that the random movement of cells is prevalent.
		\begin{figure}[ht!]
			\centering
			\begin{tabular}{c}
				{\quad\includegraphics[width=0.35\textwidth]{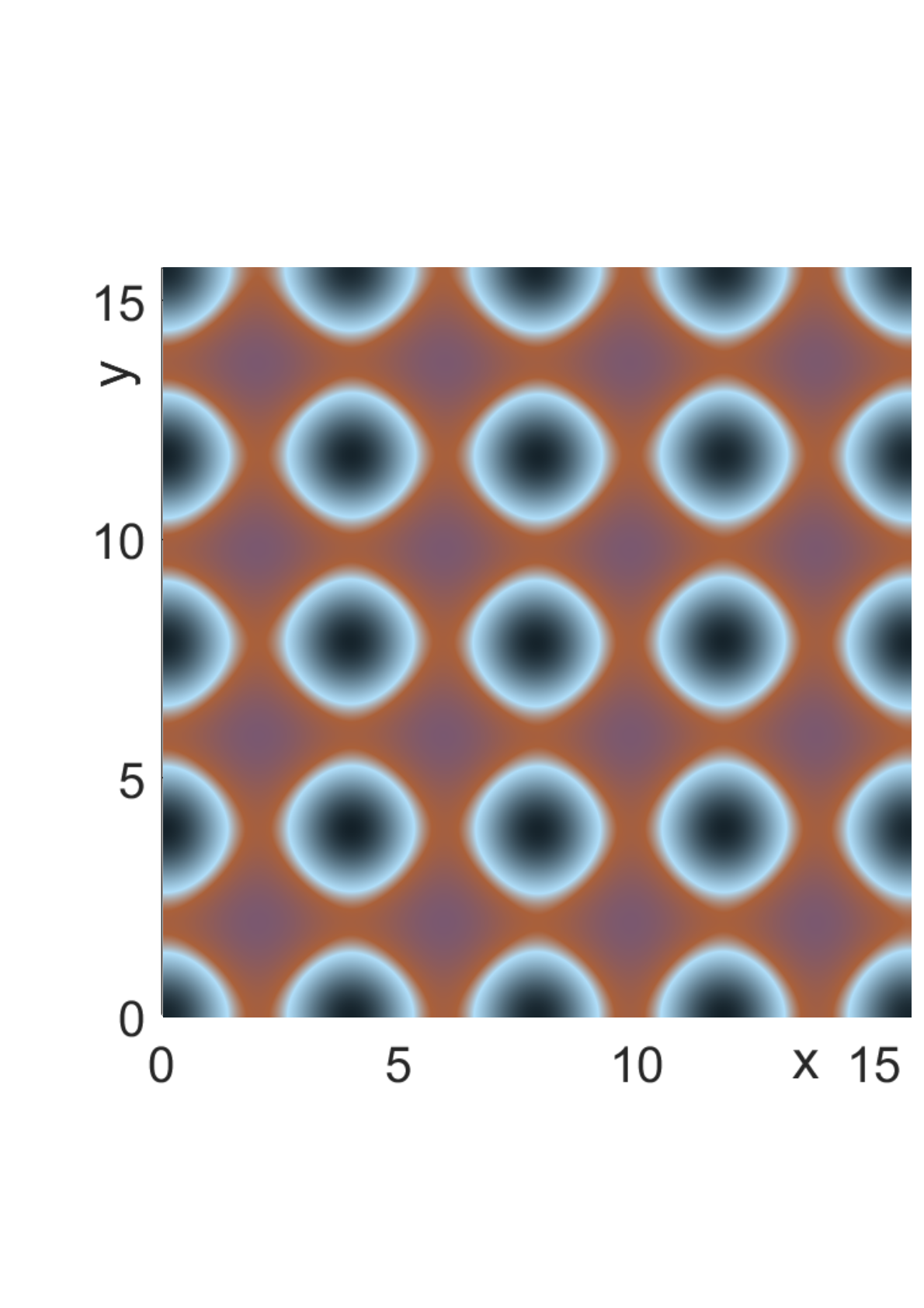}
					\qquad	\includegraphics[width=0.4\textwidth]{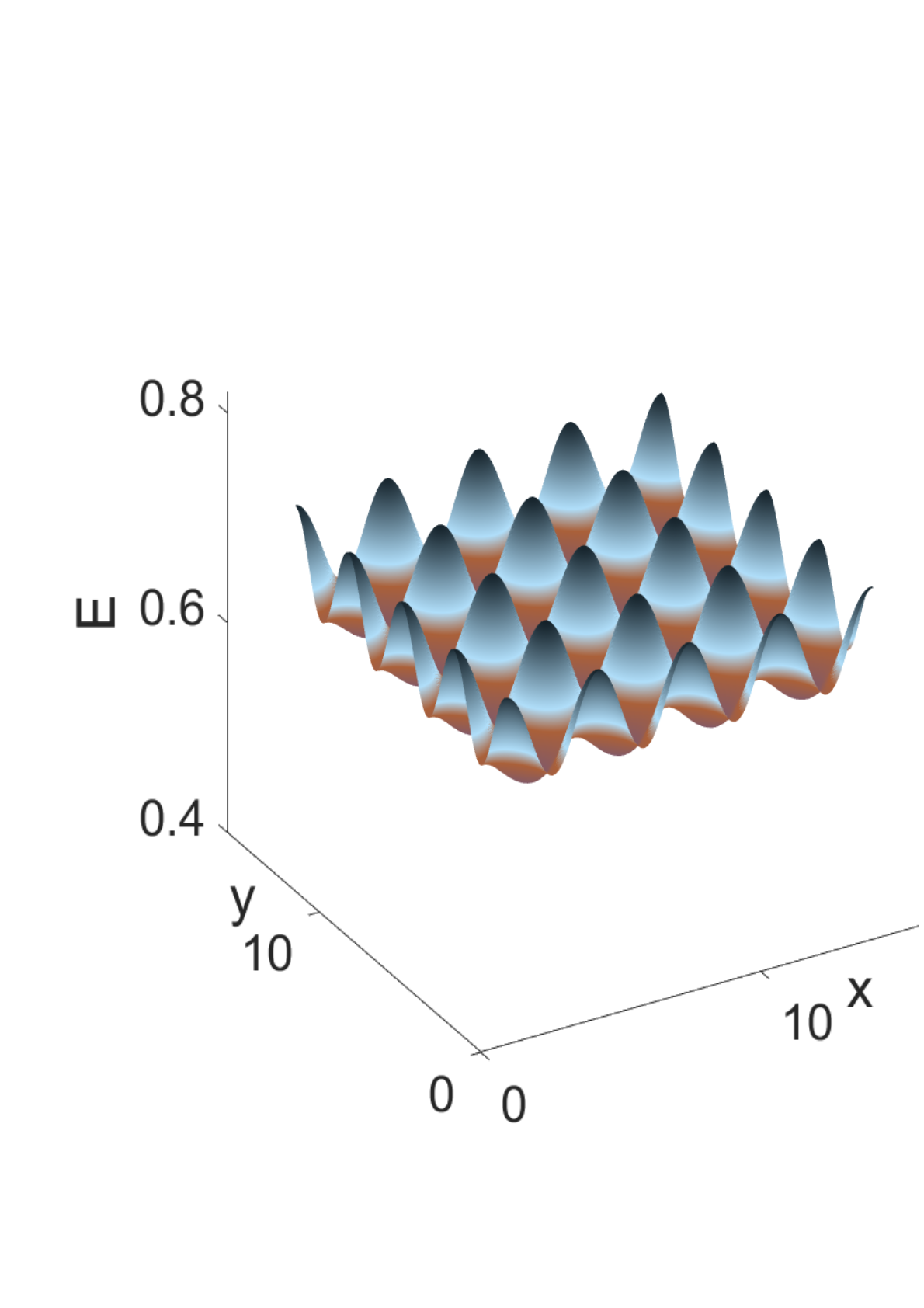}}\\
				\includegraphics[width=0.6\textwidth]{Barra.png}
			\end{tabular}
			\caption{ Long-time patterning of destroyed myelin portion ($E$), described by system \eqref{eq:rdsR}, taking parameters  $\gamma=2$, $\delta=0.6$ and  $\mathcal R=3.2$, $\xi=10.88$. }\label{Cases2}
		\end{figure}
		
		As a further test, we take the same parameters as the reference one, but we increase the parameter $\xi_m=0.1$; in this way, the chemotactic sensitivity results  $\xi\simeq 14.5$. The obtained configuration, as shown in Figure \ref{Cases3}, is still striped one. The observable difference is in the minimum and maximum levels of consumed myelin, which in this case are, respectively, lower and higher w.r.t. the reference case. This is a consequence of the fact that the motion of self-reactive immune cells is majorly driven by the chemotactic attraction of cytokines, causing an inflammation, and consequently a myelin injury, more concentrated in some areas of the domain.
		\begin{figure}[ht!]
			\centering
			\begin{tabular}{c}
				{\includegraphics[width=0.35\textwidth]{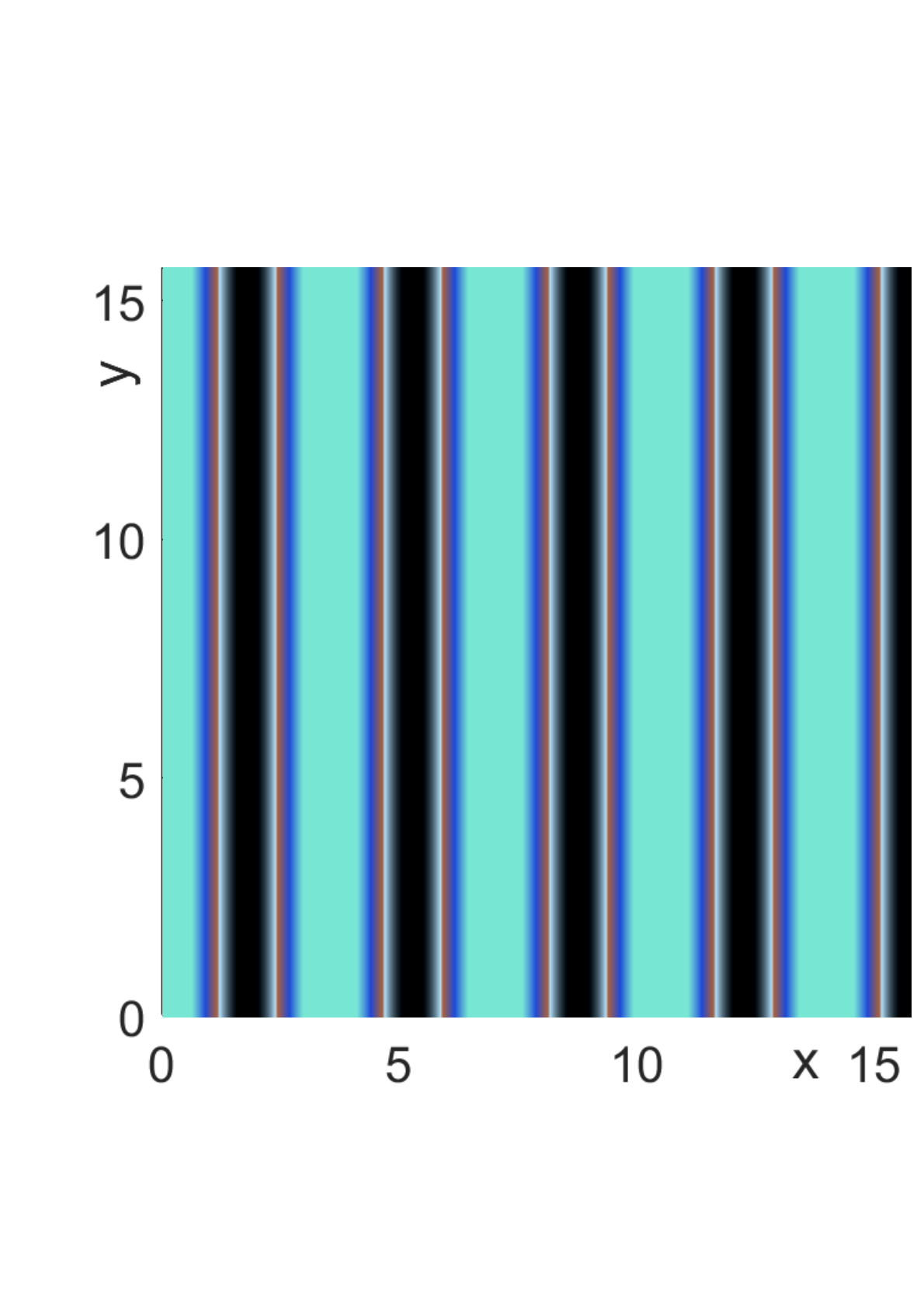}
					\qquad\includegraphics[width=0.4\textwidth]{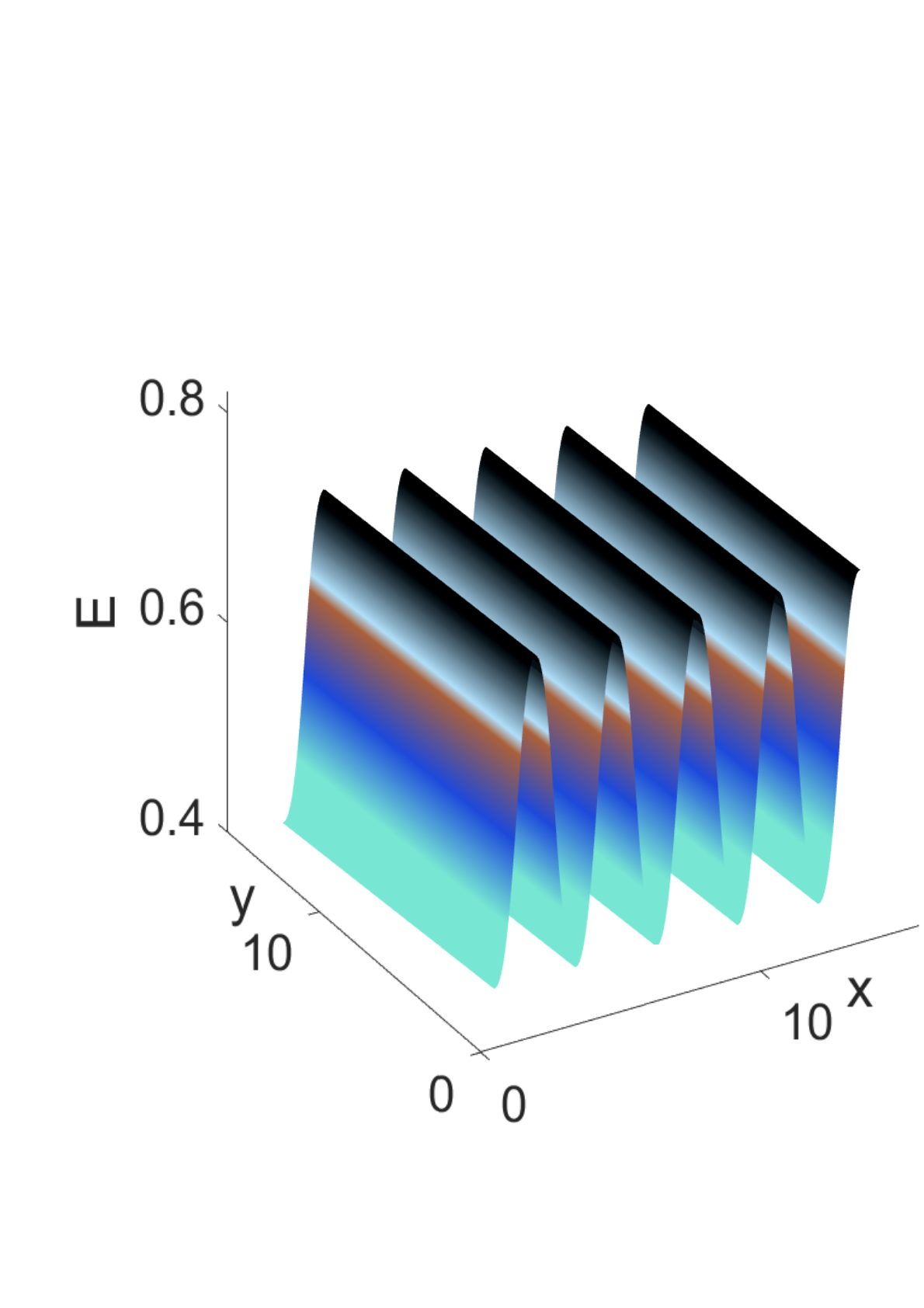}}\\
				\includegraphics[width=0.6\textwidth]{Barra.png}
			\end{tabular}
			\caption{ Long-time patterning of destroyed myelin portion ($E$), described by system \eqref{eq:rdsR} , taking parameters $\gamma=2$, $\delta=0.6$ and  $\mathcal R=3.2$,  $\xi=14.5$. }\label{Cases3}
		\end{figure}
		
		We point out that the patterns depicted above have been obtained by starting from a perturbation of the equilibrium resembling the results given in Subsection \ref{SubWeak}. Nevertheless, we underline that a weakly nonlinear analysis can be performed by decomposing the solution around the critical bifurcation value into three active dominant pairs of eigenmodes (i.e., as in the expression \eqref{SolAj}, with $m=3$), as already proposed for other models describing Multiple Sclerosis (see, e.g., \cite{bisi2025derivation}).  In this case, the solution is shaped by the combination of three amplitudes $\rho_1,\rho_2,\rho_3$ that can lead to the emergence of striped ($\rho_1\neq0,\rho_2=\rho_3=0$), hexagonal ($\rho_1=\rho_2=\rho_3\neq0$) or rectangular ($0\neq\rho_1=\rho_2\neq\rho_3\neq0$) patterns. Leaving this kind of analysis as a future step, we also perform numerical simulations for our model by starting from a random perturbation of the equilibrium. In particular, we take parameters' values $\gamma=1$, $\delta=0.6$, $\mathcal R=3.2$, and perform a test by taking $\xi=13.31$ (as in the reference case) and another taking $\xi=14.5$
		(as in the third case considered). The results obtained are reported in Figure \ref{Cases4} and show the coexistence of other stable spatial configurations, which show rectangular  (panel (a)) or hexagonal   (panel (b)) symmetry.
		\begin{figure}[ht!]
			\centering
			\begin{tabular}{cc}
				\quad(a)&\quad (b)\vspace{-0.8cm}\\	\includegraphics[width=0.3\textwidth]{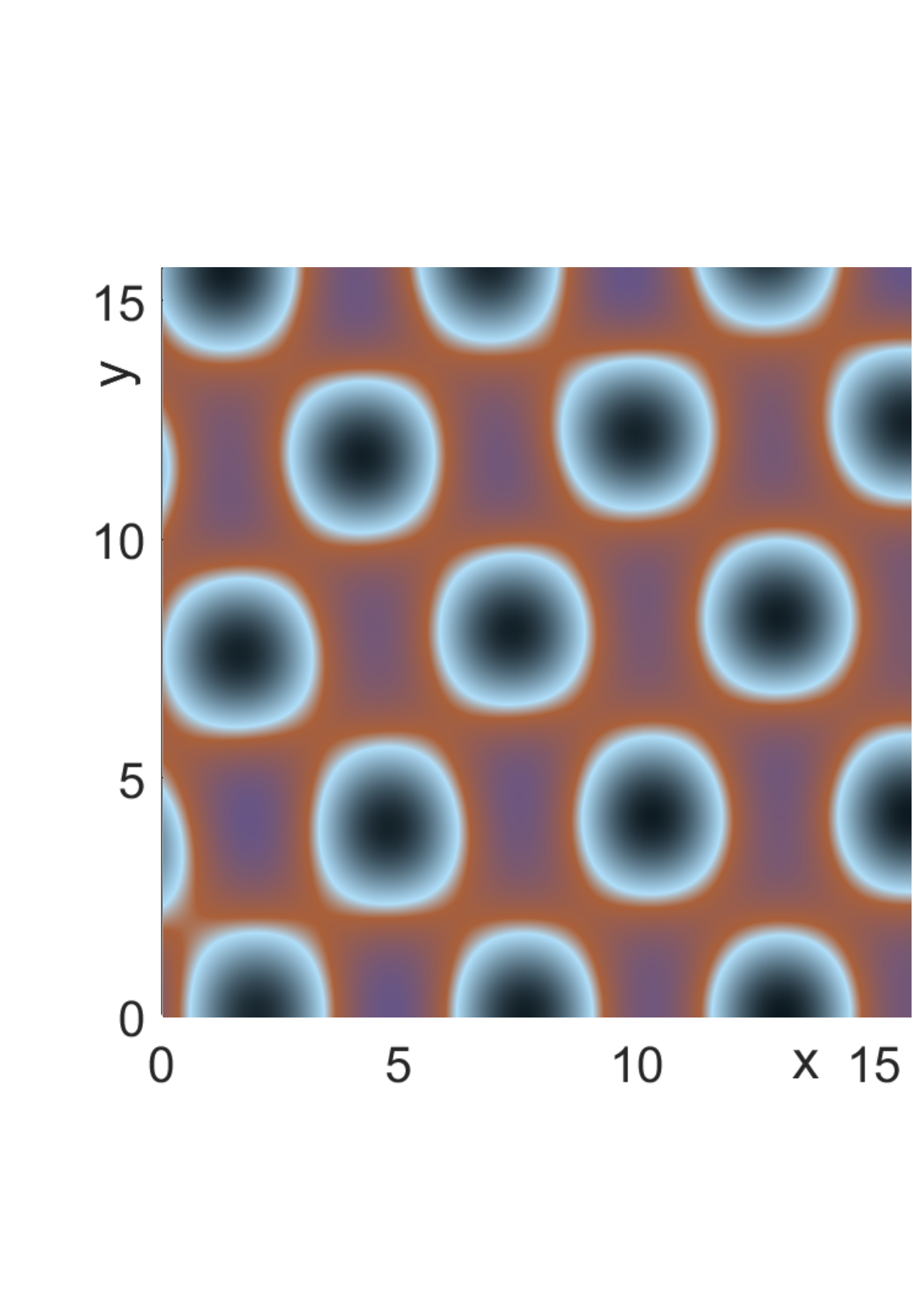}
				&	\includegraphics[width=0.3\textwidth]{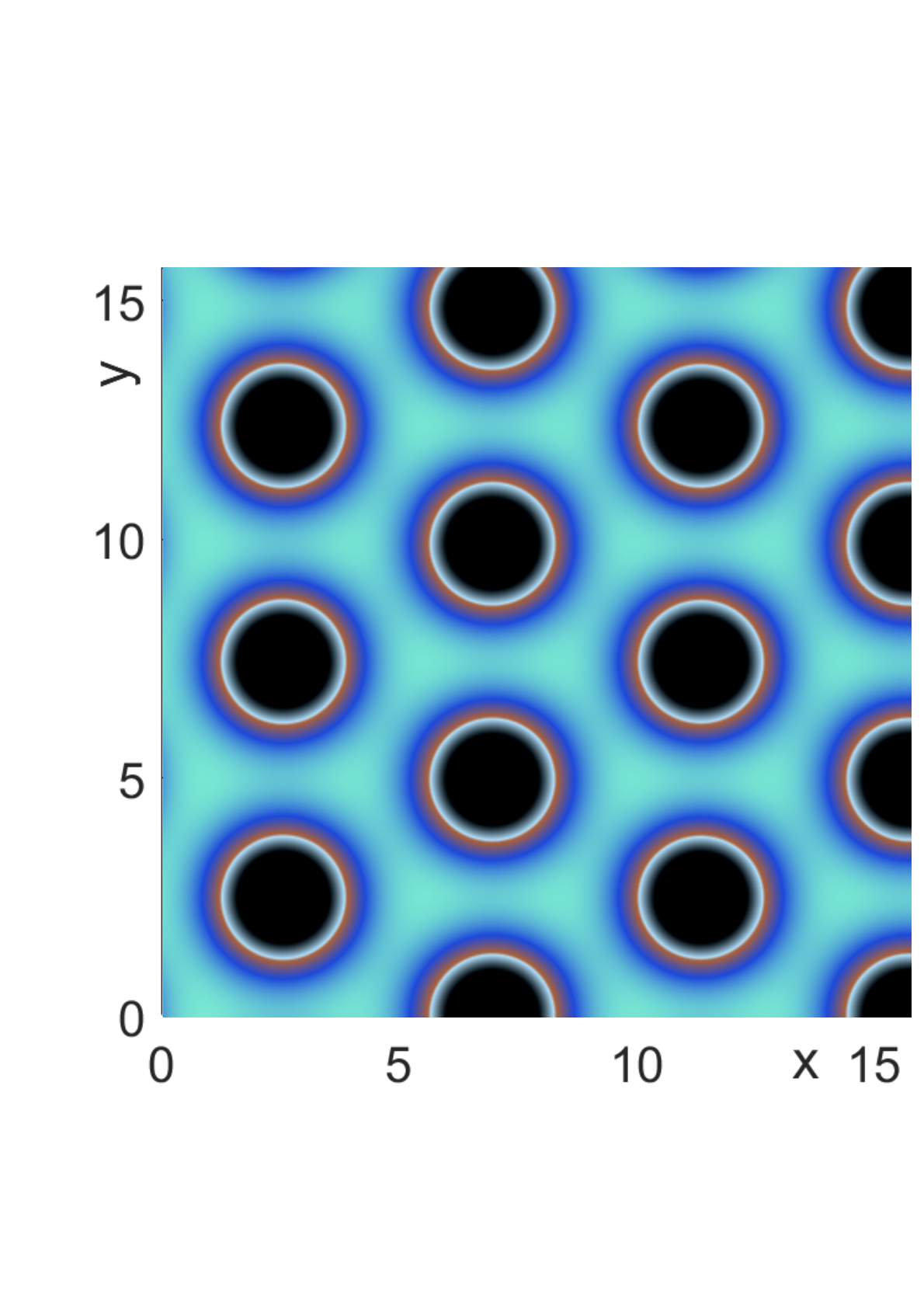}
			\end{tabular}\vspace{-0.8cm}
			\\	\qquad\includegraphics[width=0.6\textwidth]{Barra.png}
			\caption{ Long-time patterning of destroyed myelin portion ($E$), described by system \eqref{eq:rdsR} starting from a random perturbation of equilibrium, with parameters $\gamma=1$, $\delta=0.6$ and  $\mathcal R=3.2$.  Panel (a): $\xi=13.31$.   Panel (b):   $\xi=14.5$. }
			\label{Cases4}
		\end{figure}
		\section{Conclusions}
		\label{sec:5}
		In this study, we consider a reaction-diffusion model for Multiple Sclerosis and investigate the formation of two-dimensional spatial patterns that may replicate the development of demyelinating lesions in the white matter caused by the disease.
		
		To describe the dynamics of self-reactive immune cells, cytokines, and the destroyed portion of myelin, we perform a reduction of the system introduced in paper \cite{travaglini2023reaction}. We conduct a Turing instability analysis to determine parameter conditions that allow for the formation of spatial patterns. Subsequently, we carry out a weakly nonlinear analysis of the problem to further explore different patterning outcomes near the bifurcation value. Specifically, we decompose the solution in terms of two pairs of dominant eigenmodes.
		
		We then identify the conditions that may lead to the emergence of striped or squared-like patterns and analyze their stability over time. The results obtained are numerically validated, revealing that different choices of the squeezing probability, which regulates the motion of immune cells in space, can lead to different pattern shapes. Furthermore, we observe that a higher value of the chemotactic parameter results in myelin damage more concentrated in specific regions of the domain.
		
		Finally, we perform numerical simulations starting from random initial conditions, which lead to the formation of patterns exhibiting hexagonal or rectangular symmetry.  A novel contribution of this work is that, through a combination of weakly nonlinear analysis and numerical simulations, we demonstrate how key parameters, including the squeezing probability and chemotactic sensitivity, can generate these two-dimensional patterns, providing a direct mechanistic explanation for the emergence of diverse lesion morphologies. These insights allow a better understanding of atypical lesion formation in Multiple Sclerosis, capturing features that cannot be reproduced by one-dimensional models or by approaches relying solely on oligodendrocyte loss.
		
		These findings highlight the need for further analysis in future works, particularly considering the number of dominant eigenmodes larger than two. In addition, a wavefront invasion analysis of the problem could be performed to describe the spatio-temporal evolution of patterns.
		
		\subsection*{Acknowledgments}
		We are thankful to the anonymous reviewers for their useful comments and suggestions.
		This work was performed in the framework of activities sponsored by the Italian National Group of
		Mathematical Physics (GNFM-INdAM) and by the University of Parma (Italy) and of Naples Federico II (Italy). 
		RT is a post-doc fellow supported by the National Institute of Advanced Mathematics (INdAM), Italy. 
		RDM and RT also thank the support of the University of Parma through the action Bando di Ateneo per la ricerca 2022, co-funded by MUR-Italian Ministry of Universities and Research - D.M. 737/2021 - PNR - PNRR - NextGenerationEU  `Collective and self-organised dynamics: kinetic and network approaches'. 
		The work  was carried out in the frame of activities sponsored by the COST Action CA18232, by the Portuguese Projects UIDB/00013/2020 (\url{https://doi.org/10.54499/UIDB/00013/2020}), UIDP/00013/2020 of CMAT-UM (\url{https://doi.org/10.54499/UIDP/00013/2020}), 
		and by the Portuguese national funds (OE), through the Project FCT/MCTES  PTDC/03091/2022, 
		``Mathematical Modelling of Multi-scale Control Systems: applications to human diseases -- CoSysM3''
		(\url{https://doi.org/10.54499/2022.03091.PTDC}).
		%
		
		\appendix	
		\section{Proof of results in Subsection \ref{SubWeak}}\label{App}
		\subsection{Proof of Proposition \ref{Prop2}}\label{App1}
		As stated above, we are interested in describing spatial configurations arising from two pairs of eigenvectors. For this reason, we look for a solution of \eqref{SistOrd1} in the form
		\be\label{expu1sol}
		\left(
		\begin{array}{c}u_1\\v_1\\w_1 \end{array}
		\right)=\sum_{j=1}^2\left[{\mathbf B}_j(t)\,e^{i\,{\bf k_j\cdot\bx}}+\overline{\mathbf B}_j(t)\,e^{-i\,{\bf k_j\cdot\bx}}\right],
		\ee
		i.e., we express the solution in terms  of two active dominant pairs of eigenmodes
		$({\bf k}_j , -{\bf k}_j )$, $j=1,2$, individuating angles of $\pi/2$,  with $|{\bf k}_j| = k_c$ \cite{Ouyang2000,walgraef2012spatio}.
		Being \eqref{expu1sol} solution of \eqref{SistOrd1}, we have
		$$
		0=\mathcal L_c\,\left(
		\begin{array}{c}u_1\\v_1\\w_1 \end{array}
		\right)=\left(\mathbb J-k_c^2\,\mathbb D_c\right)\,\left(
		\begin{array}{c}u_1\\v_1\\w_1 \end{array}
		\right).
		$$
		As shown in Subsection \ref{subsec:2.2}, though, one has Det$\left(\mathbb J-k_c^2\,\mathbb D_c\right)=0$, and, since,
		$$
		\mathbb J-k_c^2\,\mathbb D_c=
		\begin{pmatrix}
			-1 - k_c^2\, \tilde{\Phi}_0
			& k_c^2\, \xi_c\, \tilde{\Phi}_1 
			& 0 \\[3mm]
			\beta 
			& -k_c^2\, \delta- \tau
			& 0 \\[3mm]
			\dfrac{ \zeta\, \theta\, (1 + 2\Omega) }{ (1 + \Omega)\, (\zeta + \theta + \zeta\, \Omega) } 
			& 0 
			& -\dfrac{ \zeta + \theta + \zeta\, \Omega }{ 1 + \Omega }
		\end{pmatrix},
		$$
		we can observe that the last two rows of the matrix are linearly independent, thus the kernel of $\mathbb J-k_c^2\,\mathbb D_c$ has dimension $1$.  
		Consequently, the solution \eqref{expu1sol} is of the form
		$$
		\left(
		\begin{array}{c}u_1\\ v_1\\ w_1 \end{array}
		\right)=\left(
		\begin{array}{c}\rho_R\\1\\\rho_E \end{array}
		\right)\sum_{j=1}^2\left[{\mathcal W}_j(t)\,e^{i\,{\bf k_j\cdot\bx}}+\overline{\mathcal W}_j(t)\,e^{-i\,{\bf k_j\cdot\bx}}\right],
		$$
		and, recalling the expression for $k_c^2$ in \eqref{CappaCri}, the expressions for $\rho_R$ and $\rho_E$ read as in \eqref{rhoR} and \eqref{rhoE}, respectively. This completes the proof.
		\begin{flushright}$\square$\end{flushright}
		
		\subsection{Proof of Proposition \ref{Prop3}}\label{App2}
		First, we recast \eqref{EqOrd2} as 
		\be\label{EqOrd2.2}
		\mathcal L_c
		\left(
		\begin{array}{c}u_2 \\ v_2\\w_2 \end{array}
		\right)
		= 
		\frac{\pa}{\pa T_1}\left(
		\begin{array}{c}u_1 \\ v_1\\w_1 \end{array}
		\right)
		- \mathcal H_2\left[\left(
		\begin{array}{c}u_1 \\ v_1\\w_1 \end{array}
		\right)\right].
		\ee
		We now exploit the Fredholm solvability condition, which ensures that, being $B$ a Fredholm operator on a Banach (or Hilbert) space and $b$ an element in the space, the inhomogeneous equation $Bx = b$ admits at least one solution if and only if $b$ is orthogonal to every element in the kernel of the adjoint operator $B^+$; in other words, $Bx = b$ is solvable if and only if  $\left\langle b, v \right\rangle = 0$ for all $v \in \ker(B^+)$.
		In this particular case, the operator $\mathcal L_c$ may be interpreted as a linear continuous operator from the Banach space $\mathcal Z=\left(H^2(\Gamma_{\bf x})\right)^3$, generated by the functions $e^{i\,{\tilde{\bf k}\cdot\bx}}$, $\tilde{\bf k}\in \mathbb{R}^2$, to $\mathcal X=\left(L^2(\Gamma_{\bf x})\right)^3$, with the inner product defined as the classical product in $L^2$, see \cite{da2002second,haragus2010local}  for the details. Consequently, the solvability condition states that, for a nontrivial solution to exist, the right-hand side of \eqref{EqOrd2.2} must be orthogonal to the kernel of the adjoint operator of $\mathcal L_c$, which we indicate by $\mathcal L_c^+$ and reads 
		$$
		\mathcal L_c^+  = \begin{pmatrix}
			-1 -\tilde k^2\, \tilde{\Phi}_0
			& \beta
			& 0 \\[3mm]
			\tilde k^2\, \xi_c\, \tilde{\Phi}_1 
			&  -\tilde k^2\, \delta- \tau
			& \quad\dfrac{ \zeta\, \theta\, (1 + 2\Omega) }{ (1 + \Omega)\, (\zeta + \theta + \zeta\, \Omega) }  \\[3mm]
			0
			& 0 
			& -\dfrac{ \zeta + \theta + \zeta\, \Omega }{ 1 + \Omega }
		\end{pmatrix},
		$$ being $\tilde k=|\tilde{\bf k}|$, and whose kernel is spanned by 
		\be { \label{sigmaVec}
			\left(
			\begin{array}{c}1\\\sigma_C\\0 \end{array}
			\right)e^{ i\,{\bf k\cdot\bx}}+\mbox{c.c.},\quad |{\bf k}|=k_c,\quad}	
		\sigma_C=\dfrac{1}{\beta} \left(1 + \sqrt{\dfrac{\tau \, \Phi_0}{\delta}}\right).\ee
		
		At this point, we substitute \eqref{Expu1} into the right-hand side of \eqref{EqOrd2.2}, which results in a linear combination of terms  $e^0$, $e^{i\,\bk_j\cdot\bx}$, $e^{2\,i\,\bk_j\cdot\bx}$, $e^{\,i\,(\bk_j-\bk_l)\cdot\bx}$. By isolating the coefficients corresponding to  $e^{i\,\bk_j\cdot\bx}$, they write, for $ j=1,2,$
		\be\label{Rj}
		\begin{aligned}
			\left(
			\begin{array}{c}R_u^j \\[2mm] R_v^j\\[2mm]R_w^j \end{array}
			\right)&=\left(
			\begin{array}{c}\rho_R\\1\\\rho_E \end{array}
			\right)\frac{\pa \mathcal W_j}{\pa T_1}-\left(
			\begin{array}{c}			\xi_1\,\tilde{\Phi}_1\,k_c^2\,\mathcal W_j
				\\0\\0  \end{array}
			\right).
		\end{aligned}
		\ee

		The solvability condition implies the terms  $\begin{pmatrix}
			R_u^j,  R_v^j ,  R_w^j
		\end{pmatrix}^T e^{i\,{\bf k_{\it j}\cdot\bx}} $ to be orthogonal to \eqref{sigmaVec} and thus\\
		$$\left\langle
		\begin{pmatrix}
			R_u^j \\[2mm]  R_v^j \\[2mm]  R_w^j
		\end{pmatrix}
		,
		\begin{pmatrix}
			1 \\[2mm]  \sigma_C \\[2mm]  0
		\end{pmatrix}\right\rangle 
		= 0,\quad \mbox{ for } j=1,2.$$ Consequently, by looking at \eqref{Rj}, we obtain \eqref{Solv1}. 
		
		Successively, we have that the complete expression on the right-hand side of \eqref{EqOrd2.2} results 
		\ben
		\begin{aligned}
			& \left(
			\begin{array}{c}2\,\rho_R^2 \\[1mm] 0\\[1mm]
				\Psi_0
			\end{array}
			\right)\left(|\mathcal W_1|^2+|\mathcal W_2|^2\right)+\sum_{j=1}^2 \left(
			\begin{array}{c}R_u^j \\[2mm] R_u^j\\[2mm]R_w^j \end{array}
			\right)\,\,e^{i\,\bk_j\cdot\bx}\\[2mm]
			&\quad+\left(
			\begin{array}{c}2 \, \rho_R \left( \rho_R + k_c^2 \left( \rho_R \tilde{\Phi}_0' - \xi_c \rho_c \tilde{\Phi}_1' \right) \right)
				\\[1mm] 0\\[1mm]
				\Psi_1
			\end{array}
			\right)\\
			&\quad\times\left[
			\mathcal W_1\, \mathcal W_2\,e^{i\,(\bk_1+\bk_2)\cdot\bx}
			+\mathcal W_1\,\overline{ \mathcal W}_2\,e^{i\,(\bk_1-\bk_2)\cdot\bx}        
			\right]\\
			&\quad+\sum_{j=1}^2
			\left(
			\begin{array}{c}\,\rho_R^2+2\,k_c^2\left(\rho_R^2\,\tilde{\Phi}_0'-\xi_c\,\rho_R\,\tilde{\Phi}_1'\right)
				\\[1mm] 0\\[1mm] \Psi_2
			\end{array}
			\right)\,\mathcal W_j^2\,e^{2\,i\,\bk_j\cdot\bx}+\mbox{ c.c.},
		\end{aligned}
		\ee
		with
		$$
		\Psi_0=	\dfrac{4 \, \theta \, \rho_R}{(1 + \Omega)^2} 
		\left( 
		\rho_E \, (1 + 2 \, \Omega) + \rho_R \, \left(1 + \dfrac{\theta}{\zeta} - \Omega - \dfrac{(\zeta + \theta)^2}{\zeta \, (\theta + \zeta \, (1 + \Omega))} \right)
		\right),   
		$$
		$$
		\Psi_1=\dfrac{4 \, \Theta \, \rho_R}{(1 + \Omega)^2}      \left( \rho_E (1 + 2 \Omega) - \dfrac{\zeta \, \rho_R \, \Omega^2}{\Theta + \zeta (1 + \Omega)} \right),
		$$
		$$
		\Psi_2=\dfrac{2 \, \theta \, \rho_R}{(1 + \Omega)^2} 
		\left( 
		\rho_E \, (1 + 2 \, \Omega) - \dfrac{\zeta \, \rho_R \, \Omega^2}{\theta + \zeta \, (1 + \Omega)} 
		\right).
		$$
		At this point, we claim that the solution $(u_2, v_2,w_2)^T$ for the non-homogeneous problem  \eqref{EqOrd2.2} has to be in the form 
		$$
		\begin{aligned}
			\left(
			\begin{array}{c}u_2 \\ v_2\\w_2 \end{array}
			\right)
			=& \left(
			\begin{array}{c}X_0 \\ Y_0\\Z_0 \end{array}
			\right)\left(|\mathcal W_1|^2+|\mathcal W_2|^2\right)\\
			&+\sum_{j=1}^2\left[ \left(
			\begin{array}{c}\rho_R \\ 1\\\rho_E \end{array}
			\right)\,\mathcal V_j+\mathcal{L}c^{-1}\left(
			\begin{array}{c}R_U^j \\[2mm] R_V^j\\[2mm]R_W^j \end{array}
			\right)\,\mathcal W_j\right]\,e^{i\,\bk_j\cdot\bx}
			\\
			&+ \left(
			\begin{array}{c}X_{1} \\ Y_{1}\\Z_{1}\end{array}
			\right)\,\left[\mathcal W_1\, \mathcal W_2\,e^{i\,(\bk_1+\bk_2)\cdot\bx}
			+\mathcal W_1\,\overline{ \mathcal W}_2\,e^{i\,(\bk_1-\bk_2)\cdot\bx}\right]
			\\
			&+\sum_{j=1}^2 \left(
			\begin{array}{c}X_{2} \\ Y_{2}\\Z_{2}\end{array}
			\right)\,\mathcal W_j^2\,e^{2\,i\,\bk_j\cdot\bx}+\mbox{ c.c.}.
		\end{aligned}
		$$
		Then, the left-hand side of \eqref{EqOrd2.2} reads
		$$
		\begin{aligned}
			\mathcal L_c&\left(
			\begin{array}{c}u_2 \\ v_2\\w_2 \end{array}
			\right)
			= \mathbb J \left(
			\begin{array}{c}X_0 \\ Y_0\\Z_0 \end{array}
			\right)\left(|\mathcal W_1|^2+|\mathcal W_2|^2\right)
			\\&+\left( \mathbb J-2\,k_c^2\mathbb D_c\right)\left(
			\begin{array}{c}X_{1} \\ Y_{1}\\Z_{1}\end{array}
			\right)
			\left[
			\mathcal W_1\, \mathcal W_2\,e^{i\,(\bk_1+\bk_2)\cdot\bx}
			+\mathcal W_1\,\overline{ \mathcal W}_2\,e^{i\,(\bk_1-\bk_2)\cdot\bx}        
			\right]\\
			&+\left( \mathbb J-4k_c^2\mathbb D_c\right)\sum_{j=1}^2 \left(
			\begin{array}{c}X_{2} \\ Y_{2}\\Z_{2}\end{array}
			\right)\,\mathcal W_j^2\,e^{2\,i\,\bk_j\cdot\bx}+\mbox{ c.c.},
		\end{aligned}
		$$			
		where we have removed the inhomogeneous part in the $e^{i\,\bk_j\cdot\bx}$-coefficient to avoid secular terms.
		To recover the solution of \eqref{EqOrd2.2}, then, 
		we must solve the linear equations having as unknown the coefficients $X_m, Y_m, Z_m$, $m=0,1,2$, given by
		$$
		\left(
		\begin{array}{c}X_0 \\ Y_0\\Z_0 \end{array}
		\right)=
		\mathbb J\,^{-1}
		\left(
		\begin{array}{c}2\,\rho_R^2 \\[1mm] 0\\[1mm]
			\Psi_0   
		\end{array}
		\right),
		$$
		$$
		\left(
		\begin{array}{c}X_1 \\ Y_1\\Z_1 \end{array}
		\right)
		=
		\left(
		\mathbb J-2k_c^2\mathbb D_c\right)^{-1}
		\left(
		\begin{array}{c}2 \, \rho_R \left( \rho_R + k_c^2 \left( \rho_R \tilde{\Phi}_0' - \xi_c \rho_c \tilde{\Phi}_1' \right) \right)
			\\[1mm] 0\\[1mm]
			\Psi_1
		\end{array}
		\right),
		$$
		$$
		\left(
		\begin{array}{c}X_2 \\ Y_2\\Z_2 \end{array}
		\right)
		= \left(
		\mathbb J-4k_c^2\mathbb D_c\right)^{-1}
		\left(
		\begin{array}{c}\,\rho_R^2+2\,k_c^2\left(\rho_R^2\,\tilde{\Phi}_0'-\xi_c\,\rho_R\,\tilde{\Phi}_1'\right)
			\\ 0\\ \Psi_2
		\end{array}
		\right).$$
		The resolution of the three linear equations above provides the expressions for the coefficients given in \eqref{CoefX0}-\eqref{CoefX2}, 
		with
		$$
		\Xi_0=\frac{\zeta\, \theta\, \left(\zeta + 3\, \zeta\, \Omega + \theta\, (3 + 8\, \Omega)\right)}{\left(\zeta + \theta + \zeta\, \Omega\right)^3}
		,
		$$
		$$
		\begin{aligned}
			\Xi_1=&\frac{
				2\, \left(\sqrt{\delta} + 2\, \sqrt{\tau\,\tilde{\Phi}_0}\right)}{
				\sqrt{\delta}\, \left(\sqrt{\delta\,\tilde{\Phi}_0}\, (\tilde{\Phi}_1 - \tilde{\Phi}_1') + \sqrt{\tau}\, (\tilde{\Phi}_0'\, \tilde{\Phi}_1 - \tilde{\Phi}_0\, \tilde{\Phi}_1')\right)\, \left(\zeta + \theta + \zeta\, \Omega\right)^3
			}\\&\times
			\left[
			\sqrt{\tau}\, \left(\tilde{\Phi}_0'\, \tilde{\Phi}_1 - \tilde{\Phi}_0\, \tilde{\Phi}_1'\right)\, (1 + \Omega)\, \left(\zeta + \theta + \zeta\, \Omega\right)^2\right.\\&\left.
			+ \sqrt{\delta\,\tilde{\Phi}_0}\, \left(
			\theta^2\, (\tilde{\Phi}_1 - \tilde{\Phi}_1')\, (1 + \Omega)\right.\right.\\&\left.\left.
			- 2\, \zeta\, \theta\, \left(\tilde{\Phi}_1'\, (1 + \Omega)^2 + \tilde{\Phi}_1\, \left(\theta + 3\, \theta\, \Omega - (1 + \Omega)^2\right)\right)\right.\right.\\&\left.\left.
			+ \zeta^2\, \left(\tilde{\Phi}_1 - \tilde{\Phi}_1'\, (1 + \Omega)^3 + \tilde{\Phi}_1\, \Omega\, \left(3 + \Omega\, (3 + 2\, \theta + \Omega)\right)\right)
			\right)
			\right]
			,
		\end{aligned}
		$$
		$$
		\begin{aligned}
			\Xi_2=&\frac{-
				\zeta\, \theta\, \sqrt{\tau}}{
				\sqrt{\tilde{\Phi}_0}\,
				\left(
				-\sqrt{\delta\,\tilde{\Phi}_0}\, (\tilde{\Phi}_1 - 2\, \tilde{\Phi}_1')
				+ 2\, \sqrt{\tau}\, \left(-\tilde{\Phi}_0'\, \tilde{\Phi}_1 + \tilde{\Phi}_0\, \tilde{\Phi}_1'\right)
				\right)\,
				\left(\zeta + \theta + \zeta\, \Omega\right)^3
			}\,\\&\times
			\left[
			4\, \delta\, \sqrt{\tilde{\Phi}_0}\, (\tilde{\Phi}_1 - 2\, \tilde{\Phi}_1')\, (1 + 2\, \Omega)\, \left(\zeta + \theta + \zeta\, \Omega\right)\right.\\&\left.
			+ 2\, \tau\, \sqrt{\tilde{\Phi}_0}\, (\tilde{\Phi}_0'\, \tilde{\Phi}_1 - \tilde{\Phi}_0\, \tilde{\Phi}_1')\, (1 + 2\, \Omega)\, \left(\zeta + \theta + \zeta\, \Omega\right)\right.\\&\left.
			+ \sqrt{\delta\,\tau}\, \left(
			8\, \theta\, \tilde{\Phi}_0'\, \tilde{\Phi}_1\, (1 + 2\, \Omega)
			+ 8\, \zeta\, \tilde{\Phi}_0'\, \tilde{\Phi}_1\, (1 + \Omega)\, (1 + 2\, \Omega)\right.\right.\\&\left.\left.
			+ \theta\, \tilde{\Phi}_0\, \tilde{\Phi}_1\, (19 + 56\, \Omega)
			- 10\, \theta\, \tilde{\Phi}_0\, (\tilde{\Phi}_1' + 2\, \tilde{\Phi}_1'\, \Omega)\right.\right.\\&\left.\left.
			+ \zeta\, \tilde{\Phi}_0\, \left(\tilde{\Phi}_1 + \tilde{\Phi}_1\, (3 - 16\, \Omega)\, \Omega - 10\, \tilde{\Phi}_1'\, (1 + \Omega)\, (1 + 2\, \Omega)\right)
			\right)
			\right]
			.
		\end{aligned}
		$$
		completing the proof.
		\begin{flushright}$\square$\end{flushright}
		\subsection{Proof of Proposition \ref{Prop4}}\label{App3}  
		We start by rewriting \eqref{EqOrd3} as
		\ben
		\mathcal L_c
		\left(
		\begin{array}{c}u_3 \\ v_3\\w_3 \end{array}
		\right)=
		\frac{\pa}{\pa T_1}\left(
		\begin{array}{c}u_2 \\ v_2\\w_2 \end{array}
		\right)+	
		\frac{\pa}{\pa T_2}\left(
		\begin{array}{c}u_1 \\ v_1\\w_1 \end{array}
		\right)- \mathcal H_3\left[\left(
		\begin{array}{c}u_1 \\ v_1\\w_1 \end{array}
		\right),\left(
		\begin{array}{c}u_2 \\ v_2\\w_2 \end{array}
		\right)\right].
		\ee
		As done before, we rely on the Fredholm solvability condition. We substitute the expressions \eqref{Expu1} and \eqref{Expu2} in the right-hand term of \eqref{EqOrd3}, and isolate the  \eqref{Solv1}, that results
		\ben
		\begin{aligned}
			\left(
			\begin{array}{c}S_u^j \\[2mm] S_v^j\\[2mm]S_w^j \end{array}
			\right)&=\left(
			\begin{array}{c}\rho_R\\1\\\rho_E \end{array}
			\right)\left(\frac{\partial \mathcal V_j}{\partial T_1} + \frac{\partial \mathcal W_j}{\partial T_2}\right)-\left(
			\begin{array}{c}			\tilde{\Phi}_1\,k_c^2\,\left(\xi_1\,\mathcal V_j+\xi_2\,\mathcal W_j\right)
				\\0\\0  \end{array}
			\right)\\& +
			\left[\left(\begin{array}{c}-r_1\\0\\t_1  \end{array}\right)
			\,|\mathcal W_j|^2 + \left(\begin{array}{c}-r_2\\0\\t_2  \end{array}\right)\,|\mathcal W_l|^2 \right]\,\mathcal W_j,                     
		\end{aligned}
		\ee
		for $j,l=1,2$, $j\neq l$, with coefficients $r_1$ and $r_2$ given in \eqref{r1r2} and
		\ben
		\begin{aligned}
			t_1 &= q_1 \left( q_2 + (Z_0 + Z_2)\, q_3 + (X_0 + X_2)\, q_4 \right), \\
			t_{2} &= q_1 \left( 2\, q_2 + (Z_0 + 2\, Z_1)\, q_3 + (X_0 + 2\, X_1)\, q_4 \right), \\
			q_1 &= \frac{2\, \theta\, \rho_R}{(1 + \Omega)^3\, \left(\zeta + \theta + \zeta\, \Omega\right)^2},\quad 
			q_2 = 9\, \zeta\, \rho_R^2\, \Omega^2\, \left(\zeta + \theta + \zeta\, \Omega\right), \\
			q_3 &= (1 + \Omega)\, (1 + 2\, \Omega)\, \left(\theta + \zeta\, (1 + \Omega)\right)^2, \\
			q_4 &= \zeta\, (1 + \Omega)\, \left(\theta - 2\, \zeta\, \Omega^2\, (1 + \Omega) + 2\, \theta\, \Omega\, (2 + \Omega)\right).
		\end{aligned}
		\ee
		At this point, imposing $$\left\langle
		\begin{pmatrix}
			S_u^j \\[2mm]  S_v^j \\[2mm]  S_w^j
		\end{pmatrix}
		,
		\begin{pmatrix}
			1 \\[2mm]  \sigma_C \\[2mm]  0
		\end{pmatrix}\right\rangle 
		= 0,\quad \mbox{ for } j=1,2,$$
		with $\sigma_C$ defined in \eqref{sigmaVec}, we find the relation \eqref{Solv2}, completing the proof.
		\begin{flushright}$\square$\end{flushright}
		\subsection{Proof of Proposition \ref{Prop5}} \label{App4}
		First, equations \eqref{Solv1} - \eqref{Solv2} allow us to rewrite \eqref{dAjdt}, for  $j,l=1,2$, $l\neq j,$, as
		\ben
		\begin{aligned}
			\left(\rho_R + \sigma_C\right)\frac{d {A}_j}{d t}=&k_c^2\,\tilde{\Phi}_1\left(\mathcal W_j\left(\eta^2\,\xi_1\,+\eta^3\,\xi_2\right)+ \eta^3\,\xi_1\,\mathcal V_j\right)\\&+\eta^3\,\left(r_1\,|\mathcal W_j|^2 + r_2\,|\mathcal W_l|^2 \right)\,\mathcal W_j,
		\end{aligned}
		\ee 
		omitting the higher-order terms.
		
		From \eqref{expxi} we may write, still omitting higher-order terms, $$\eta\, \xi_1=\xi-\xi_c-\eta^2\xi_2,$$
		obtaining 
		\ben
		\begin{aligned}
			\left(\rho_R + \sigma_C\right)\frac{d {A}_j}{d t}=&k_c^2\,\tilde{\Phi}_1\left(\xi-\xi_c\right)\left(\eta\,\mathcal W_j+ \eta^2\,\xi_1\,\mathcal V_j\right)\\&+\eta^3\,\left(r_1\,|\mathcal W_j|^2 + r_2\,|\mathcal W_l|^2 \right)\,\mathcal W_j,
		\end{aligned}
		\ee
		which, from \eqref{ExpAj}, becomes
		\be\label{SistAj}
		s_0\frac{d A_j}{d t}=\xi_m \,A_j+A_j\,\left(s_1\,|{A_j}|^2+s_2\,|{A_l}|^2\right),
		\ee
		where 
		\ben s_0=\dfrac{\rho_R+\sigma_C}{k_c^2\,\xi_c\,\tilde{\Phi}_1},\quad \xi_m=\dfrac{\xi-\xi_c}{\xi_c},\quad s_i=\dfrac{r_i}{k_c^2\,\xi_c\,\tilde{\Phi}_1}, \ee
		with $i=1,\,2$.
		
		By decomposing each amplitude into modulus and phase angle as $$A_j=\rho_j\,e^{i\,\phi_j},$$ and splitting the real and imaginary parts of equations \eqref{SistAj}, we get the system:
		\ben
		\begin{aligned}
			s_0\left(\rho_1\, \frac{d\phi_1}{dt}\, \cos{\phi_1} + \frac{d\rho_1}{dt}\,\sin{\phi_1}\right)
			&= \rho_1\, \left( \xi_m + s_1\, \rho_1^2 + s_2\, \rho_2^2 \right)\, \sin{\phi_1}, \\[6pt]
			s_0\left(\rho_2\, \frac{d\phi_2}{dt}\,\, \cos{\phi_2} + \frac{d\rho_2}{dt}\, \sin{\phi_2}\right)
			&= \rho_2\, \left( \xi_m + s_1\, \rho_2^2 + s_2\, \rho_1^2 \right)\, \sin{\phi_2}, \\[6pt]
			s_0\left(\frac{d\rho_1}{dt}\, \cos{\phi_1} - \rho_1\, \frac{d\phi_1}{dt}\, \sin{\phi_1}\right)
			&= \rho_1\, \left( \xi_m + s_1\, \rho_1^2 + s_2\, \rho_2^2 \right)\, \cos{\phi_1}, \\[6pt]
			s_0\left(\frac{d\rho_2}{dt}\, \cos{\phi_2} - \rho_2\, \frac{d\phi_2}{dt}\, \sin{\phi_2}\right)
			&= \rho_2\, \left( \xi_m + s_1\, \rho_2^2 + s_2\, \rho_1^2 \right)\, \cos{\phi_2},
		\end{aligned}
		\ee
		that provides
		\ben
		\begin{aligned}
			\frac{d\phi_1}{dt}&=0\\
			\frac{d\phi_2}{dt}&=0\\
			s_0\,\frac{d \rho_1}{d t} &= \xi_m\,\rho_1 +s_1 \, \rho_1^3 +s_2\,\rho_2^2\, \rho_1
			\\
			s_0\,\frac{d \rho_2}{d t} &= \xi_m \,\rho_2 +s_1 \, \rho_2^3 +s_2\,\rho_1^2  \, \rho_2,
		\end{aligned}
		\ee
		completing the proof.
		\begin{flushright}$\square$\end{flushright}
		
		\bibliographystyle{abbrv}
		\bibliography{biblioAbb}
		%
		%
		%
		%
		%
	\end{document}